\documentclass[twoside,leqno]{article}
\usepackage[letterpaper]{geometry}
 
\usepackage{ltexpprt}
\newtheorem{observation}{Observation}[section]
\newtheorem{claim}{Claim}[section]
\usepackage[hyperfootnotes=false]{hyperref}

\usepackage[utf8]{inputenc}
\usepackage{amsmath}
\usepackage{amsfonts}
\usepackage[dvipsnames]{xcolor}
\usepackage{xspace}
\usepackage{graphicx}
\usepackage{comment}
\usepackage{subcaption}
\usepackage{thm-restate}
\usepackage{enumitem}

\usepackage{cleveref}
\crefformat{observation}{Observation #2#1#3}
\crefformat{@theorem}{Theorem #2#1#3}
\crefformat{claim}{Claim #2#1#3}

\newcommand{\R}{{\ensuremath{\mathbb{R}}}\xspace}
\newcommand{\A}{{\ensuremath{\mathcal{A}}}\xspace}

\newcommand{\seg}{{\ensuremath{\textrm{seg}}}\xspace}
\newcommand{\homo}{{\ensuremath{\simeq}}\xspace}

\newcommand{\ER}{{\ensuremath{\exists \mathbb{R}}}\xspace}
\newcommand{\mnev}{Mn\"{e}v\xspace}

\newcommand{\satobject}{union of hypercube faces}
\newcommand{\satobjects}{unions of hypercube faces}

\begin{document}

\title{\Large Topological Art in Simple Galleries}
\author{Daniel Bertschinger\thanks{Department of Computer Science, ETH Z\"{u}rich, Switzerland.} \and Nicolas El Maalouly$^*$ \and Tillmann Miltzow\thanks{Deparment of Information and Computing Sciences, Utrecht University, Netherlands.} \and Patrick Schnider\thanks{Department of Mathematical Sciences, University of Copenhagen, Denmark.} \and Simon Weber$^*$}
\date{}

 \maketitle

\begin{abstract} \small\baselineskip=9pt
    Let $P$ be a simple polygon, then the art gallery problem is looking for a minimum
    set of points (guards) that can see every point in~$P$.
    We say two points~$a,b\in P$ can see each other if the line segment $\seg(a,b)$ 
    is contained in~$P$.
    We denote by~$V(P)$ the family of all minimum guard placements.
    The Hausdorff distance makes $V(P)$ a metric space and thus a topological space. 
    We show homotopy-universality, that is
    for every semi-algebraic set $S$ there is a polygon $P$ such that $V(P)$  is homotopy equivalent to~$S$.

    Furthermore, for various concrete topological spaces $T$, we describe instances $I$ of the art gallery problem such that $V(I)$ is homeomorphic to $T$.
\end{abstract}

\vfill

\begin{center}
    \includegraphics{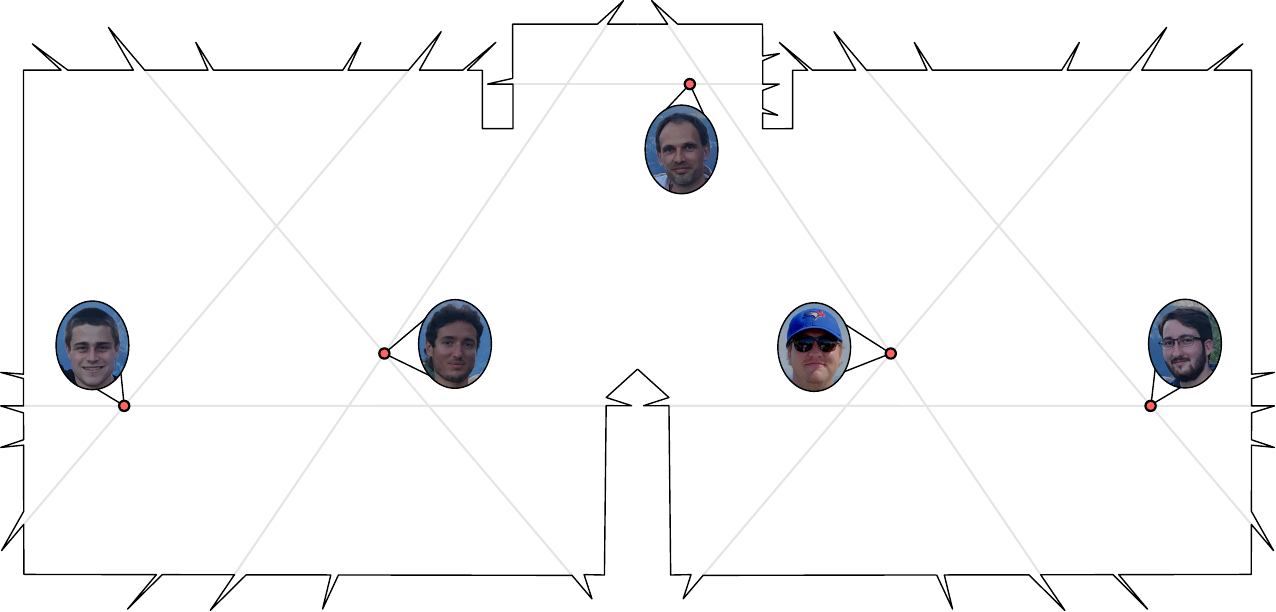} 

    \vspace{0.2cm}
    
    {The five authors jointly guarding this polygon. The space of such optimal guard placements forms a topological $2$-chain, i.e., two cycles connected by an edge.}
\end{center}

\vfill

\pagebreak
\section{Introduction}
\label{sec:intro}
The art gallery problem is a fundamental problem in computational geometry. In this work we study the topological spaces formed by the set of all optimal guard placements. 
We show that a large variety of spaces can be realized.

See \Cref{sec:preliminaries} for a detailed definition of the art gallery problem, its solution space~$V(I)$, its topology, semi-algebraic sets, homeomorphism, homotopy equivalence, and the existential theory of the reals.

\paragraph{Homotopy-Universality.}
Given a set of points, the order type of this point set
describes for each ordered triple whether they are oriented clockwise or counter-clockwise.
Imagine we are given two point sets $P_1$ and $P_2$ with the same order type.
Can we continuously move $P_1$ to $P_2$ without changing the order type?
For most people, intuition says that the answer is always yes, as this seems to be the case 
for small examples. 
Let~$O$ be the order type of $P_1,P_2$, we denote the realization space by~$V(O)$.
The question above can be rephrased to whether $V(O)$ is connected.
\mnev{}'s seminal universality theorem states that for every semi-algebraic set~$S$ there exists 
an order type $O$ such that $S$ and $V(O)$ are 
homotopy equivalent~\cite{bjorner1999oriented,knuth1992axioms,matousek2014intersection,mnev1988universality,shor1991stretchability}.
Specifically, if $S$ is the set that consists of two isolated points, $V(O)$ is also disconnected.
Thus there are indeed two point sets in the plane with the same order type that cannot be transformed continuously into one another.
\mnev{} inspired a deep study of order types, which opened up a new research field.
In particular, \mnev's universality result is one of the most significant contributions
to the \textit{existential theory of the reals}.
The existential theory of the reals is both an algorithmic problem (ETR) and an algorithmic complexity class (\ER).
Based on \mnev{}s original ideas it was shown that order-type realizability 
is complete for this class~\cite{matousek2014intersection,shor1991stretchability}.
Order-type realizability is 
central, as it captures the combinatorics
of the simplest geometric objects: points in the plane.
Furthermore, the \ER-hardness of 
many other algorithmic problems relies on it~\cite{bienstock1991some,cardinal2017intersection,SlopeNumber,mcdiarmid2013integer}.

\paragraph{Link to \ER-completeness.}
 Homotopy-universality and \ER-completeness  are closely linked. 
In an \ER-hardness reduction we transform semi-algebraic sets $S$,
given by polynomials, to instances $I$ of some other problem.
Again we denote by $V(I)$ its solution space.
In order to show \ER-hardness, the transformation has to be in polynomial time, and $S =  \emptyset \Leftrightarrow V(I) = \emptyset$.
Note that the topology of $V(I)$ may be widely different to the topology of $S$.
Similarly, in order to show homotopy-universality, 
we need to transform a semi-algebraic set into some instance $I$.
In this setting, we have no restriction on the running time of the transformation.
The transformations do not even need to be computable to show homotopy-universality.
Instead, 
we need to 
show that $S$ and $V(I)$ are homotopy equivalent.
It may not be so surprising that a transformation of the first type may sometimes be massaged to a transformation of the second type and vice versa.

\paragraph{Art Gallery Problem.}
The art gallery problem is another central problem in computational geometry.
The aim in the art gallery problem is to place as few points~$G$ (guards) as possible
into a given polygon~$P$ such that every point $q\in P$ is seen by a point $g\in G$.
The art gallery problem was invented by Viktor Klee~\cite{o1987art} and is since then
one of the central problems in computational geometry.
Let us mention a few milestones.
1978, Steve Fisk wrote a beautiful one-page proof that $\lfloor \frac{n}{3} \rfloor$ guards
are always sufficient to guard a polygon on $n$ vertices~\cite{chvatal1975combinatorial,Fisk78a}.
1986, Lee and Lin~\cite{LeeLin86} first established NP-hardness, see also~\cite{eidenbenz2001inapproximability,SchuchardtH95}.
1987, Joseph O'Rourke wrote a beautiful book \textit{Art~Gallery Theorems and Algorithms}~\cite{o1987art}, 
summarizing the state of the art at the time.
Surprisingly, the art gallery problem is probably not contained in NP, as it is \ER-complete~\cite{ARTETR}.
In recent years, the art gallery problem has been studied also from various other theoretical perspectives: approximation~\cite{BonnetM17Approx, EfratH02}, smoothed analysis~\cite{SmoothedArt,SmoothingGap}, and parameterized complexity~\cite{AlmostConvex,agrawal2020parameterized,ashok2019efficient,BonnetW1HARD}.
From a practical perspective, various research groups developed, implemented, and tested algorithms that are exact, reliable, and relatively fast, but theoretically may run forever~\cite{PracticalBorrmann,PracticalARTMasterFriedrich,Simon-hengeveld2021practical,tozoni}.

For us the \ER-completeness of the art gallery problem is most relevant. 
Usually, \ER-hardness proofs imply homotopy-universality theorems as well.
Although it is plausible that such a statement follows from the \ER-hardness reduction it is not formally stated in \cite{ARTETR}.
Instead, a slightly different statement about semi-algebraic sets in the plane is formulated~\cite{ARTETR}.
Still, it leaves an inconvenient gap in our understanding of the art gallery problem.

\paragraph{Results.}
We fill this gap by providing a simple and short proof of homotopy-universality in the context of the art gallery problem.

Our result even holds for a version of the art gallery problem where only the vertices of the polygon have to be guarded. Note that the guards can still be placed anywhere in the polygon. Following the language of Agrawal and Zehavi \cite{agrawal2020parameterized}, we call this problem version the \textsc{Point-Vertex Art Gallery}. 
This problem is contained in NP, for a short proof see \Cref{app:NPcontainment}.
With this notation the ordinary art gallery problem is denoted as the \textsc{Point-Point Art Gallery}
and vertex guarding is denoted as \textsc{Vertex-Point Art Gallery}.

\begin{restatable}[Homotopy-Universality]{@theorem}{Universal}
    \label{thm:universal}
     For every compact semi-algebraic set~$S\neq \emptyset$ there is an instance $I$ of the art gallery problem such that $V(I)$ is homotopy equivalent to~$S$. Furthermore, this also holds when $I$ is considered as an instance of \textsc{Point-Vertex Art Gallery}.
\end{restatable}

\noindent The proof can be found in \Cref{sec:Realization}
There are two observations that may also hold for other algorithmic problems.
    First, homotopy-universality may be generally easier to show than \ER-completeness.
    While a reduction often gives both \ER-hardness and homotopy-universality after 
    sufficiently massaging it, 
    it might be easier to just give two separate proofs that 
    are each more accessible individually. 
    Second, to show homotopy-universality, 
    we only use that the vertices of the polygon
    must be seen. 
    All other interior points are seen automatically in our construction.
    As mentioned above, this version of the art gallery problem lies in NP.
    Thus homotopy-universality is already achievable for NP-complete problems. 

    Based on a preliminary version of this paper, Stade and Tucker-Foltz~\cite{stade2022topological} nicely extended our findings to homeomorphism-universality, that is, they showed that for any compact semi-algebraic set there exists an art gallery with a solution space homeomorphic to that set. 
    However, their results use the fact that the walls of the polygon have to be guarded as well (but not the interior), and thus do not imply homeomorphism-universality of \textsc{Point-Vertex Art Gallery}. The homeomorphism-universality of this problem version remains an open problem.

Our construction for homotopy-universality does not yield homeomorphism-universality, as the dimension of the solution space gets blown up by the freedom of movement of guards we use to enforce homotopy equivalence. To add on to \Cref{thm:universal},  we provide a sequence of tangible examples that have solution spaces homeomorphic
to standard topological spaces, even in \textsc{Point-Vertex Art Gallery}.
See \Cref{fig:Overview_Tangible} for an overview of these examples. 

\begin{restatable}[Tangible $\&$ Homeomorphic]{@theorem}{Tangible}
    \label{thm:tangible}
     For the following topological spaces~$S$ there exists an instance~$I$ of the art gallery problem such that $V(I)$ is homeomorphic to~$S$. For any $k\geq 1$,
     \begin{center}
         1. $k$-clover; \quad 
         2. $k$-chain; \quad 
         3. $4k$-necklace; \quad 
         4. $k$-sphere; \quad 
         5. torus and double torus. 
     \end{center}
\end{restatable}

\begin{figure}[ht]
\centering
\includegraphics[page = 3]{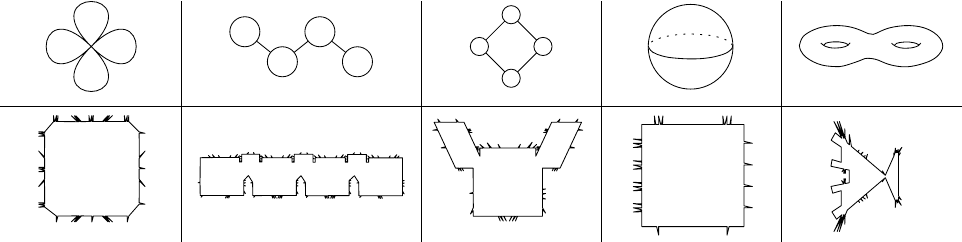}
\caption{Left to right: $4$-clover, $4$-chain, $4$-necklace, $2$-sphere, double torus.
}
\label{fig:Overview_Tangible}
\end{figure}

The proofs can be found in \Cref{sec:tangible}.
To help the imagination with our more involved constructions, we created \textrm{GeoGebra} applications that show the correspondence between guard placements and points in the topological space. 
An accompanying video showcases the most interesting parts.

\begin{center}
    \url{https://www.youtube.com/watch?v=xWBmzB0bhso}
\end{center}

Our tangible examples nicely illustrate the topological complexity that can be achieved with very small polygons and very few guards.
We believe that our results are suitable as classroom material both in Computational Geometry as well as in Algebraic Topology.
Terms like homotopy equivalence and homeomorphism can be explained and their relevance can be illustrated.
The explicit description of our examples, coupled with the fact that the spaces of optimal solutions are very different to typical benchmark instances,
gives motivation to use our polygons as test cases for art gallery algorithms~\cite{abrahamsen2017irrational,Simon-hengeveld2021practical}.
Note that even simple to construct polygons easily lead to very complicated topological solution spaces.
The selection we present here illustrates how we can realize nice and well-known topological spaces.
We do not know whether we can find a tangible instance~$I$ of the art gallery problem such that 
$V(I)$ is homeomorphic to the projective plane or the Klein Bottle.
We find these spaces interesting as both are unorientable surfaces.

\paragraph{Proof Ideas.}
Our constructions are based on \emph{guard segments}, which are also
used in various other works~\cite{abrahamsen2017irrational,ARTETR}. 
A guard segment is a segment inside the polygon that needs to contain at least one guard.
For a simple construction of a guard segment, see \Cref{fig:guard-segment}. 

Our homotopy-universality theorem crucially relies on guard segments to encode variables.
In addition, we build gadgets that enforce that only certain
combinations of intervals on each guard segment can be used.
In this way, it is easy to construct arbitrary unions of hypercube faces.
As we will see this is sufficient to show homotopy-universality.
The technical challenge is to show that we can contract the additional dimensions our construction introduces to the desired solution space.

\begin{figure}[htbp]
    \centering
    \includegraphics{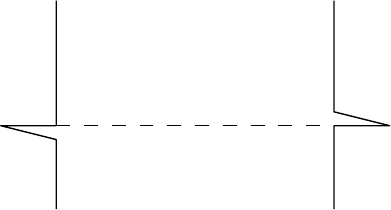}
    \caption{The pockets on the left and on the right enforce a guard to be placed on the dashed segment, if we want that one guard sees both pockets simultaneously.
    }
    \label{fig:guard-segment}
\end{figure}

In our constructions for \Cref{thm:tangible}, guard segments can be touching or crossing one another.
Thus, it is possible to use fewer guards than guard segments, and we can get interesting solution spaces. 
Intersecting guard segments also pose some technical difficulties, which will force us to use a more complex construction with more pockets. 
We will deal with these technicalities in \Cref{sec:tangible}.
We will examine the possible ``extremal'' configurations of guards on guard segments and the possibilities
to move between them. 
In this way, we can determine the structure of the solution space explicitly.
Often, very simple arrangements lead to horribly complicated topological solution spaces.
Even small changes to the construction can lead to completely different spaces.
The challenge is to be inventive to get exactly
those topological spaces that we care about.
Specifically, the double torus in \Cref{sub:doubleTorus} required a
solid understanding of the correspondence of polygons and the solution space.

\begin{figure}[ht]
\centering
\includegraphics[page =3]{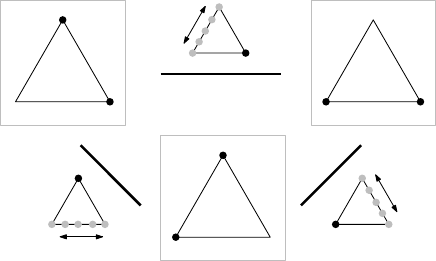}
\caption{If we can force guards onto three specially arranged guard segments forming a triangle, then every solution needs two guards, and the resulting solution space is homeomorphic to the circle.}
	\label{fig:Circle-Intro}
\end{figure}

\paragraph{Simplicity.}
We consider it a strength of our proofs that they are simple.
There are different ways to quantify simplicity.
When measured in \textit{length of the paper}, we note that the paper establishing \ER-completeness~\cite{ARTETR} has $69$~pages
and does not contain the homotopy-universality; including it would have made that paper even longer.
Our proof of the homotopy-universality is not even six pages long, with 
three pages merely being figures, running examples and additional explanations, leaving only three pages for the formal proof.
Yet, we do not consider brevity the main criterion for simplicity.

Rather, we think that \textit{ease of understanding} is the more relevant measure of simplicity.
In order to maximize it we provide {running examples}, 
we {modularize} our proofs, and
we provide {elaborated illustrations}, whenever this seems to be helpful.
Furthermore, we add {exhaustive details} and meticulously give references
and {additional explanations} whenever it might be helpful.
Each of our tangible examples is just a few pages long.
We can directly 
draw the topological space and the configurations that interest us in a single figure.
We believe that this type of integrated representation makes it particularly easy to understand, see \Cref{fig:Circle-Intro} for an example of the circle.
In particular, apart from a few technical explanations the proof
of correctness boils down to inspecting some figure.

\section{Preliminaries}
\label{sec:preliminaries}
\paragraph{Art Gallery Problem.}
Two points $p,q$ inside a polygon~$P\subset \R^2$, can see each other
if the line segment $\seg(p,q)$ is inside~$P$.
In the art gallery problem, we are given a polygon $P$ and an integer~$k$.
We are asked if there is a set $G\subset P$ 
of size $k$ such that
every point in $P$ is seen by a point in $G$. 
We call the points in $G$ \emph{guards}.
Note we consider polygons as closed 
and guards as indistinguishable.

\paragraph{Existential Theory of the Reals.}
The complexity class \ER can be defined as the set of 
algorithmic problems that is  polynomial time equivalent
 to finding real roots of polynomials $p\in \mathbb{Z}[X_1,\ldots,X_n]$.
As we will not use \ER-completeness, except for pointing out its link to homotopy-universality, we merely refer to some surveys and recent developments~\cite{abrahamsen2021training,CardinalSurvey,SmoothingGap,matousek2014intersection,Reinier-CSP,schaefer2021complexity,TensorRank, GeometricEmbeddingER}.

\paragraph{Hausdorff Distance.}
Given two points $p,q$ in the plane $\R^2$, we define the \textit{Euclidean distance} \[d(p,q) = \|p-q\|_2 = \sqrt{(p_1-q_1)^2 + (p_2-q_2)^2}.\]
Consequently, given two closed sets $F,G \subset \R^2$, we define the \textit{directed Hausdorff distance} as 
\[\overrightarrow{d_H}(F,G) = \max_{x\in F} \min_{y\in G}  d(x,y). \]
The ordinary \textit{Hausdorff distance} is defined by
\[d_H(F,G) = \max \left( \overrightarrow{d_H}(F,G) , \overrightarrow{d_H}(G,F) \right).\]

\paragraph{Solution Space.}
We describe an instance 
of the art gallery problem  by $I = (P,k)$.
We denote by $V(I)$ the set of all possible solutions, i.e.,~the set of all possible placements of $k$ guards such that $P$ is completely guarded.
When we write $V(P)$, omitting $k$, we implicitly assume that $k$ is the smallest possible value that admits at least one solution.
Note our guards are in general not distinguishable, and thus $V(I)$ is not readily embeddable into $\R^{2k}$ as each solution is a set and not a vector.
However, the Hausdorff distance still defines a topology on~$V(I)$, as every metric defines a topological space.
Note that although one could choose a different metric to the Hausdorff metric, we are not aware of any natural metric that would lead to a different topology.

In some cases, we construct instances~$I$ with disjoint guard segments such that 
every guard must be on its private guard segment.
In this case, we can interpret each guard set as a vector in $\R^k$
and the usual topology in $\R^k$ coincides with the topology given by the Hausdorff distance.
To see this simply observe that small neighborhoods are identical and these are sufficient to define a topological space.

\paragraph{Semi-Algebraic Sets.}
A \textit{basic semi-algebraic} set~$S$ is defined by polynomials with integer coefficients $(p_i)_{i=1,\ldots,k}$ and $(q_i)_{i=1,\ldots,m}$,
as follows: \[S = \{x\in \R^n : p_i(x) \geq  0 \land q_i(x) > 0\}.\]
\textit{Semi-algebraic} sets are finite unions of basic semi-algebraic sets. 

\paragraph{Cubical complexes}
A \emph{cubical complex} is a set $C$ of hypercubes that satisfy the following conditions:
\begin{enumerate}
    \item every face of a hypercube of $C$ is also in $C$;
    \item any non-empty intersection of two hypercubes $c_1$ and $c_2$ is a face of both of them.
\end{enumerate}
One example of cubical complexes that is particularly important to us are unions of (closed) faces of some hypercube.

\paragraph{Topology Terms.}
We say that two topological spaces $A,B$ are \emph{homeomorphic} if there is
a homeomorphism $\varphi \colon A \rightarrow B$.
A \emph{homeomorphism} is a bijective, continuous function, whose inverse is continuous as well.

Two maps, $f_0,f_1 : A \rightarrow B$ are \textit{homotopic}, if there is a continuous map
$H: [0,1]\times A \rightarrow B$ such that $H(0,x) = f_0(x)$ and $H(1,x) = f_1(x)$.
We call $H$ a \emph{homotopy}, and write $f_t$ for the function $H(t,\cdot)$.

Given two topological spaces $A$ and $B$, a \emph{homotopy equivalence} between $A$ and $B$ is a continuous map $f: A \rightarrow B$ for which there exists another continuous map $g : B \rightarrow A$, such that $g \circ f$ is homotopic to the identity map $\textrm{id}_A$ and $f \circ g$ is homotopic to $\textrm{id}_B$. 
If such a map exists, then $A$ and $B$ are said to be \textit{homotopy equivalent}, and we write $A\homo B$.

\section{Realization of Unions of Hypercube Faces}
\label{sec:Realization}

In this section, we aim to prove \Cref{thm:universal}.

\Universal*

To this end, we first show that for every semi-algebraic set $S$, there is a homeomorphic \satobject{} (\Cref{lem:semialgebraic}). 
Then, we show that for each \satobject{} $U$, we can build an instance of the art gallery problem, such that its solution space can be projected to a space homeomorphic to $U$ (\Cref{lem:Facial-Complex}). 
This projection is the reason for our construction only yielding homotopy equivalence and not homeomorphism. 
In our constructed polygons, every guard set guarding all vertices automatically guards the whole polygon, thus the solution space is the same when considering the two versions of the problem.

For the proofs we will need some formal definitions.
A \emph{hypercube} in dimension $n$ is the set of points $\{X\in \mathbb{R}^n:X\in [0,1]^n\}$. A \emph{face} of the hypercube is a closed subset of the hypercube given by some number of additional constraints of the form $X_i=0$ or $X_i=1$. 
A \textit{\satobject}~$U$ in dimension $n$ is the union of faces of the hypercube in dimension $n$. We call a face $F$ \emph{maximal} in $U$ if $F\not\subset F'$ for every other face $F'\subseteq U$. 

\begin{figure}[phtb]
    \centering
    \includegraphics{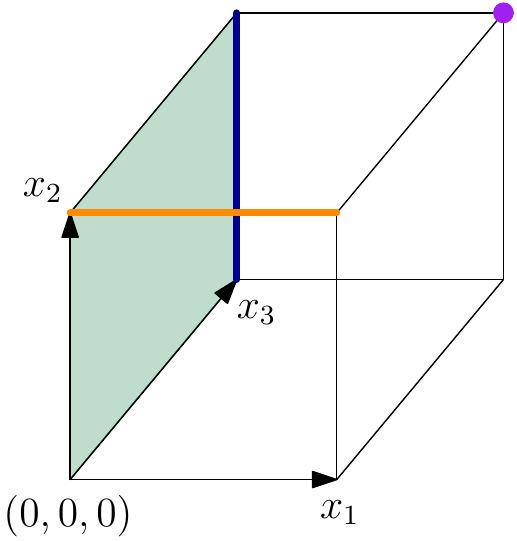}
    \caption{An example of a \textit{\satobject}~in dimension $3$ consisting of $3$ maximal faces. The blue edge is a non-maximal face, as it is a subset of the green face.}
    \label{fig:complex}
\end{figure}

\begin{lemma}\label{lem:semialgebraic}
For every compact semi-algebraic set $S$, there is a \satobject{} $U$, such that $U$ is homeomorphic to $S$.
\end{lemma}

\begin{proof}
It was shown by Hironaka that all compact semi-algebraic sets can be triangulated \cite{hironaka1975triangulations}, that is, for every such semi-algebraic set $S$ there exists a simplicial complex $T$ which is homeomorphic to $S$.
While the proof of Hironaka is phrased only for algebraic sets, it goes through for semi-algebraic sets, see e.g.~\cite{hofmann2009triangulation}.
Furthermore, each simplicial complex can be refined into a cubical complex which can be embedded into a hypercube as the union of some of its faces (\cite{blass1972cubical}, Theorem 1.1).
So in particular, there exists a \satobject{} $U$ which is homeomorphic to $T$ and thus also to $S$.
\end{proof}

Note that computing this union of hypercube faces $U$ implicitly determines whether the semi-algebraic set~$S$ is empty or non-empty. This is the underlying reason why our reduction does not yield \ER-hardness.
Furthermore, note that the description complexity of $U$ is not necessarily polynomial in the description complexity of~$S$.

The previous lemma shows that we can focus on \satobjects.

\begin{lemma}
\label{lem:Facial-Complex}
    Let~$U$ be a non-empty \satobject{} in dimension $n$. Then there exists an instance of the art gallery problem~$I$ such that $V(I)$ is homotopy equivalent to $U$. Furthermore, $V(I)$ can be projected to a set that is homeomorphic to $U$.
\end{lemma}
The mentioned projection will become apparent in the proof.

To prove this lemma, we first find a formula that encodes $U$.
We define a \textit{variable constraint} on a real variable $X_i$ as either the formula $X_i=0$ or $X_i=1$.
We define a \textit{facial conjunction} as the conjunction of variable constraints for distinct variables, i.e., 
\[C = \big((X'_{1}=v_1)\land \ldots \land (X'_{k}=v_k)\big)\]
for some distinct $X'_1,\ldots,X'_k\in\{X_1,\ldots,X_n\}$ and for $v_1,\ldots ,v_k\in \{0,1\}$.

We define a \textit{facial disjunctive formula} $\varphi$ as a disjunction of finitely many facial conjunctions, i.e., 
\[\varphi = C_1 \lor C_2 \lor \ldots \lor C_m.\]

\begin{observation}
    For every non-empty \satobject{} $U$ there exists a facial disjunctive formula~$\varphi$ such that
    $U = \{ x\in [0,1]^n : \varphi(x)\} $.
    Similarly, every facial disjunctive formula~$\varphi$ defines a \satobject{} $\{ x\in [0,1]^n : \varphi(x)\}$.
\end{observation}

For example, the union of hypercube faces in \Cref{fig:complex} is encoded by the formula
\[\varphi = (X_1 = 0) \lor (X_2 = 1 \land X_3 = 0) \lor (X_1 = 1 \land X_2 = 1 \land X_3 = 1).\]
Note that the encoding uses one facial conjunction for each maximal face, the non-maximal blue face does not have to be encoded as it is a subset of the green face.

We can now prove \Cref{lem:Facial-Complex}. Let $\varphi$ be some facial disjunctive formula that encodes~$U$ and contains~$m$ 
facial conjunctions, one for each maximal face in $U$, i.e., $\varphi = C_1\lor \ldots \lor C_m$.

The proof goes in several steps.
First, we show that $\varphi$ can be rewritten as a conjunction of disjunctions, using $m$ additional binary \emph{satisfier variables} $S_i$, of which at least one has to be false.
Then, we show how the real variables $X_i$ and binary satisfier variables $S_i$ can be encoded using guard segments, and how the encoding ensures at least one satisfier variable to be false.
Then, we show how to encode a single disjunction as a pocket in the polygon. The conjunction of these disjunctions is naturally encoded by the art gallery problem requiring the whole polygon to be guarded at the same time.
At last it remains to show the desired topological relationships between $V(P)$ and the original \satobject~$U$.

\paragraph{Rewriting Conjunctions.}
Let $C_i = \big((X'_{1}=v_1)\land \ldots \land (X'_{k}=v_k)\big)$ be a conjunction of $\varphi$.
We introduce the boolean variable $S_i$ as a satisfier variable. Intuitively, $S_i$ being true means that $C_i$ is \emph{not} the conjunction that satisfies the whole formula $\varphi$.
We now rewrite $C_i$ as 
\[E_i = \big((X'_{1} = v_1) \lor S_i\big) \land \ldots \land \big((X'_{k} = v_k) \lor S_i\big). \]
It can now be seen that $\varphi=C_1\lor\ldots\lor C_m$ is equivalent to 
$\bar{\varphi}:=E_1\land \ldots\land E_m\land (\neg S_1\lor\ldots\lor \neg S_m)$, as each $E_i$ is equivalent to $C_i\lor S_i$.
Note that $\bar{\varphi}$ is in conjunctive normal form, and all clauses except the last contain exactly two constraints.
In the following construction, this last clause will be handled separately from the others.

As an example we apply this rewriting to the example \satobject{} from \Cref{fig:complex}.
\begin{alignat*}{2}
    &&\varphi = \quad & (X_1 = 0) \lor (X_2 = 1 \land X_3 = 0) \lor (X_1 = 1 \land X_2 = 1 \land X_3 = 1) \hfill \\
    \rlap{is rewritten to}\\
    &&\bar{\varphi} = \quad & (X_1 = 0 \lor S_1) \\
    &&\land\; & (X_2 = 1 \lor S_2) \land (X_3 = 0 \lor S_2) \\
    &&\land\; & (X_1 = 1 \lor S_3) \land (X_2 = 1 \lor S_3) \land (X_3 = 1 \lor S_3) \\
    &&\land\; & (\neg S_1 \lor \neg S_2 \lor \neg S_3)
\end{alignat*}

\paragraph{Placing Variable/Dummy Guard Segments.}
\begin{figure}[tbph]
    \centering
    \includegraphics{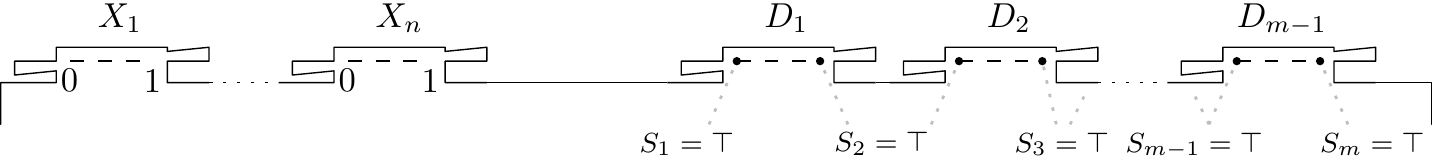}
    \caption{We define for every variable $X_i$ one variable guard segment to the left.
    Furthermore, we define $m-1$ dummy guard segments $D_i$, whose guards can satisfy at most one of two adjacent satisfier variables $S_i$ and $S_{i+1}$ by standing at either the left-most or the right-most point of the segment.
    }
    \label{fig:variables}
\end{figure}
The construction of variables can be seen in \Cref{fig:variables}.
Variables $X_i$ are realized by horizontal \emph{variable guard segments} (called $X_i$, interchangeably) at the top of the art gallery polygon, which are enforced to contain at least one guard in any optimal solution, as otherwise two guards are needed to see both pockets. The value $x_i$ of the variable $X_i$ in any optimal guard assignment is the location $x_i$ of the single \emph{variable guard} on this segment, where $0$ corresponds to the left endpoint of the segment, and $1$ to the right.

For the binary satisfier variables $S_i$, $m-1$ additional pockets with \emph{dummy guard segments} $D_1,\ldots D_{m-1}$ are added. 
In any optimal guard placement, the satisfier variable $S_i$ is considered \emph{true} ($S_i=\top)$ if the location $d_i$ of the \emph{dummy guard} on segment $D_i$ is equal to the left-most point $l(D_i)$ of the segment, or if the dummy guard on segment $D_{i-1}$ is at the right-most point $r(D_{i-1})$ of its segment. If neither is the case, the satisfier variable is considered \emph{false}.

This means that the guard on segment $D_i$ can only make either variable $S_i$ or $S_{i+1}$ true, but not both at the same time. Clearly, with only $m-1$ dummy guard segments and thus only $m-1$ dummy guards, at most $m-1$ satisfier variables are considered to be true in any optimal guard placement. 
This encodes the last clause of $\bar{\varphi}$.

\begin{figure}[htbp]
    \centering
    \includegraphics[]{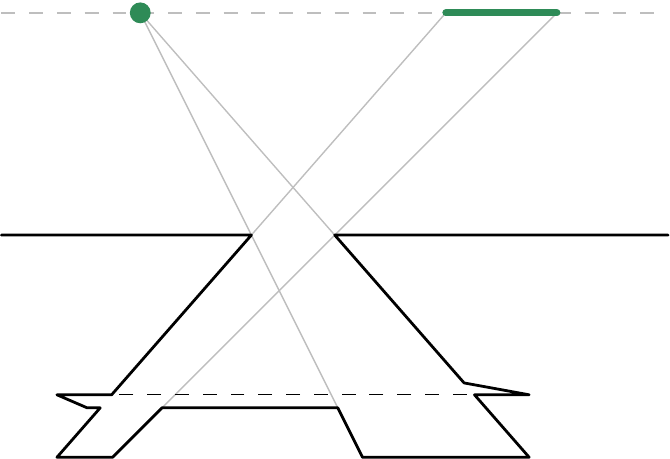}
    \caption{
    Here, we show how to encode the constraint that there is at least one guard on one of the two green intervals.
    This corresponds to a disjunction of two constraints. The helper guard that is enforced to be on the dashed segment can see into at most one of the two corridors.
    }
    \label{fig:disjunction}
\end{figure}

\paragraph{Encoding Disjunctions.}
The construction of disjunction gadgets can be seen in \Cref{fig:disjunction}.

Each disjunction of $\bar{\varphi}$, except the already handled last clause, is encoded as a pocket on the bottom side of the polygon. The entrance of the pocket functions like the pinhole on a pinhole camera, and should be sufficiently small.

Each disjunction contains one constraint of the form $X_i=v_i$ and one constraint of the form $S_j=\top$. 
The constraint $X_i=v_i$ means that the variable guard for $X_i$ is standing at the endpoint $v_i$ of its guard segment, which can be viewed as a degenerate interval $I_1:=[v_i,v_i]$. The constraint $S_j=\top$ means that at least one of the dummy guards on $D_{j-1}$ or on $D_j$ is at the correct position, which is equivalent to saying that some guard is standing on the interval $I_2:=[r(D_{j-1}),l(D_{j})]$ (for satisfier variables $S_1$ and $S_m$, this interval also degenerates to a point).

Now, for both constraints in our disjunction, we build one pocket which can only be seen through the pinhole by a guard standing anywhere on the satisfying interval. To achieve this, the left (right) point of the satisfying interval is connected with the left (right) end of the pinhole, and a sufficiently deep corridor is drawn.

To ensure that the disjunction gadget is guarded exactly if the disjunction is satisfied by any of the two constraints, one horizontal \emph{helper guard segment} $H_i$ is added above the corridors, which needs to contain at least one \emph{helper guard}. This segment is placed in such a way that the helper guard can see into at most one corridor at once from any position. We denote the position of the helper guard on segment $H_i$ by $h_i$.

\begin{observation}\label{obs:guardingdisjunction}
     The disjunction gadget can be completely guarded with exactly one guard inside of the disjunction pocket if and only if there is at least
    one guard on one of the two satisfying intervals $I_1,I_2$, assuming that all guards stand on guard segments.
\end{observation}

\paragraph{Completing the Art Gallery.}
We have defined the upper and the lower boundary of our polygon $P$. We want to ensure that the parts of these boundaries between the variable/disjunction gadgets as well as the left and right boundary can be guarded in a controlled fashion. To this end, we add one more \emph{auxiliary guard segment}. This guard segment is actually only a single point $a$ from where all of the non-interesting parts of the polygon, but no corridor of a disjunction gadget can be guarded. In any optimal solution, there will be one \emph{auxiliary guard} standing at this point.

The complete art gallery polygon corresponding to the example in \Cref{fig:complex} can be found in \Cref{fig:completepolygonexample}.

\begin{figure}[htbp]
    \centering
    \includegraphics[]{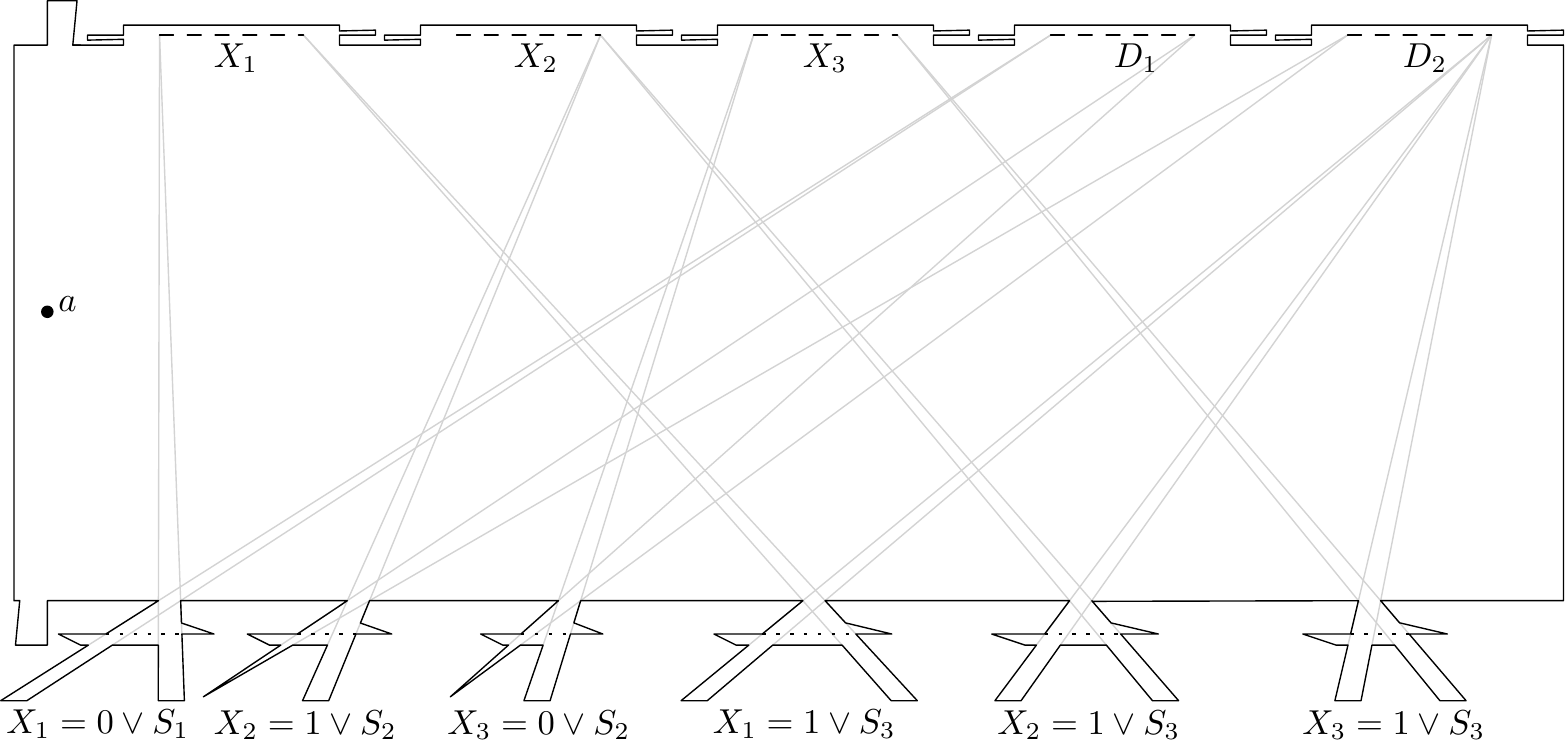}
    \caption{The complete polygon encoding the \satobject{} from \Cref{fig:complex}.}
    \label{fig:completepolygonexample}
\end{figure}

\paragraph{Solution Space.}

We will now analyze the solution space $V(P)$. First, we derive a lower bound on the number of required guards, and then an upper bound, which holds assuming $U$ is non-empty.

In the following, let $g$ denote the number of disjunction gadgets, and recall $n$ being the dimension of $U$ and $m$ being the number of conjunctions in $\varphi$.

\begin{lemma}
    At least $n+m+g$ guards are required to guard the polygon $P$.
\end{lemma}
\begin{proof}
    There are $n+m+g$ pairwise disjoint guard segments, namely the $n$ variable guard segments, the $m-1$ dummy guard segments, the $g$ helper guard segments, and the auxiliary guard segment. Because they are disjoint, at least one guard is needed for each such guard segment.
\end{proof}

\begin{lemma}\label{lem:pointtosolution}
    For any point $p\in U$, a placement of the variable guards such that $x_i=p_i$ can be extended to a solution using $n+m+g$ guards. 
\end{lemma}
\begin{proof}
    Because $p\in U$, this assignment of variables must satisfy at least one conjunction $C_j$ of $\varphi$. Therefore the formula $E_j$ is satisfied even if $S_j$ is false. If we place the dummy guards such that all other $S_k$ for $k\not=j$ are true, every disjunction is satisfied and thus every disjunction gadget has at least one of the two corridors being guarded. The helper guards can now be placed such that every corridor is guarded, and the auxiliary guard stands at its fixed point $a$, leading to a solution using $n+(m-1)+g+1$ guards.
\end{proof}

As a corollary, we get the following observation.

\begin{observation}
    If $U$ is non-empty, exactly $n+m+g$ guards are required to guard $P$. Furthermore, in every optimal solution there is exactly one guard on each guard segment, and every guard must stand on a guard segment.
\end{observation}

Now we will show the opposite direction of \Cref{lem:pointtosolution}.

\begin{lemma}\label{lem:solutiontopoint}
    Let $S$ be a placement of $n+m+g$ guards in $V(P)$. 
    Then the point $p$ with $p_i=x_i$ is contained in $U$.
\end{lemma}
\begin{proof}
    As every disjunction gadget is completely guarded in $S$, by \Cref{obs:guardingdisjunction} there must be a guard on at least one of the two intervals corresponding to assignments of variables satisfying the disjunction.
    This means that the corresponding assignment of variables and satisfier variables must satisfy $\bar{\varphi}$. 
    Therefore the assignment of variables $X_i$ must satisfy $\varphi$, and the point $p$ such that $p_i=x_i$ must be in $U$.
\end{proof}

We have therefore established a relationship between optimal guard placements and points in~$U$.

\paragraph{Projection and Homotopy Equivalence.}
We can now show the existence of the projection claimed in \Cref{lem:Facial-Complex}.
Consider any guard placement of one guard per variable, dummy, and helper guard segment, and the auxiliary guard at $a$ and consider the function $\pi$
\begin{equation}\label{eqn:projection}
\begin{aligned}
    \pi: \quad [0,1]^{n+(m-1)+g}\times\{a\} &\longrightarrow [0,1]^{n+(m-1)+g}\times\{a\} \\
    (x_1,\ldots ,x_n,d_1,\ldots ,d_{m-1},h_1,\ldots ,h_g,a) &\longmapsto(x_1,\ldots ,x_n,0\ldots ,0,0\ldots ,0,a)
\end{aligned}
\end{equation}
which maps such a placement by moving all dummy and helper guards to the left-most point of their guard segment. By definition, $\pi$ is equal to its square, therefore $\pi$ is a projection. By \Cref{lem:pointtosolution,lem:solutiontopoint}, the image $\pi(V(P))$ of the solution space under this projection is exactly the embedding of $U$ into the movement of the variable guards.
In particular, $\pi$ restricted to $V(P)$ is surjective onto $U$.
Recall that a map between topological spaces is called \textit{proper} if inverse images of compact subsets are compact.
For the map $\pi$ for any open or compact set of $U$, its pre-image is open or compact.
Thus, the projection $\pi$ is continuous and proper.

To finalize the proof of \Cref{lem:Facial-Complex}, it remains to show homotopy equivalence between $V(P)$ and $U$. 
To this end, we first show the following claim.

\begin{claim}\label{claim:contractible}
For each fixed assignment of variable guards that corresponds to a point in $U$, the set of optimal guard placements extending this assignment with helper and dummy guards is contractible.
\end{claim}
\begin{proof}
First, we show how to contract the degrees of freedom in the helper guards.
For each fixed assignment of variable guards and dummy guards, each helper guard is either completely free to move on its whole helper guard segment (if both conditions for its disjunction gadget are met by the variable and dummy guards already), or it is only free to move within the interval where it can observe the corridor that is not yet guarded by variable and dummy guards. In either situation, the space each helper guard can move on is an interval, and all helper guards are independent from each other. Therefore the space spanned by the helper guards for a fixed assignment of variable and dummy guards forms a $g$-dimensional cube, which is contractible. In particular, moving each helper guard to its leftmost possible position contracts this $g$-dimensional cube.

For any given placement of the variable guards corresponding to a point $p\in U$, $p$ can be part of one or more maximal faces. If the point is part of a single maximal face, there is a one-to-one correspondence of the point to a dummy guard assignment. If the point is part of multiple maximal faces, there is some freedom in the placement of dummy guards.

\begin{figure}
    \centering
    \includegraphics{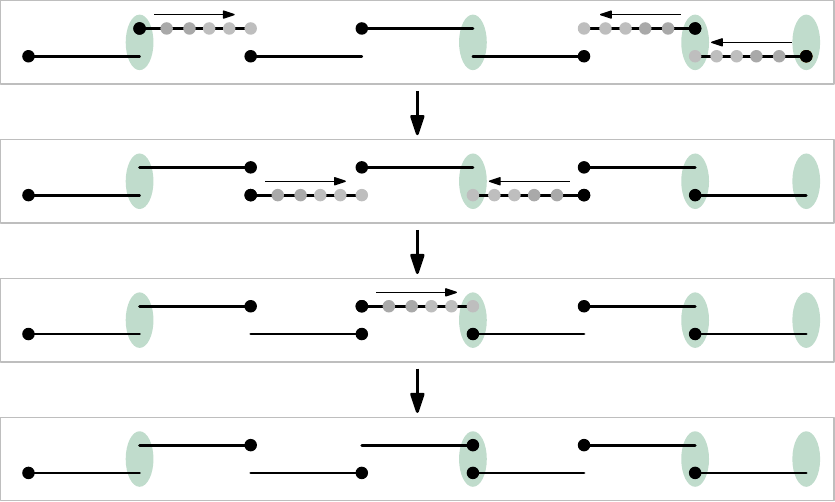}
    \caption{How dummy guards can be moved when multiple dummy variables do not need to be guarded (green ellipses). Between two adjacent such variables, guards can switch between two extremal positions: all guards to the left, or all guards to the right.
    }
    \label{fig:dummy-movement}
\end{figure}

With the point being a part of $v$ maximal faces, $v$ satisfier variables $S'_{1},...,S'_{v}\in \{S_1,\ldots ,S_m\}$ (in order) do not need to be true. In this case, between every consecutive pair $S'_{i}, S'_{i+1}$, 
the dummy guards have some freedom to move around. As all satisfier variables between $S'_{i}$ and $S'_{i+1}$ still need to be true, the possible movement of the dummy guards on the dummy guard segments between $S'_{i}$ and $S'_{i+1}$ equates to a path/interval between the two extreme solutions -- all guards to the left of their intervals, and all guard to the right of their intervals (see \Cref{fig:dummy-movement}). The dummy guard movements between each such consecutive pair are independent, therefore the solution space spanned by the dummy guards for a fixed assignment of variable guards forms a $v-1$-dimensional cube, which is also contractible. Similar to the helper guards, we can thus also move all dummy guards one-by-one to their leftmost possible position to contract this $v-1$-dimensional cube.

To summarize, moving the helper guards and the dummy guards to their leftmost possible positions define two contractions.
We have therefore shown in two steps that for fixed variable guards, the solution space spanned by the dummy and helper guards is contractible.
\end{proof}

Intuitively, we would now like to contract the solution space spanned by the dummy and helper guards to a point for each configuration of variable guards.
However, the above lemma only shows that this is locally possible, and it is not clear that these local contractions can be chosen in a way that they fit together globally.
The intuition can however be formalized using the following lemma.

\begin{lemma}[\cite{smale1957vietoris}]\label{lem:homotopy}
    Let $X$ and $Y$ be connected cubical complexes. 
    Assume that there exists a continuous proper surjective map $f:X\rightarrow Y$ such that for each point
    $y\in Y$, its pre-image $f^{-1}(y)$ is contractible. Then $X$ and $Y$ are homotopy equivalent.
\end{lemma}

\Cref{lem:homotopy} follows easily from a more general result by Smale~\cite{smale1957vietoris}.
In order to apply Smale's result
we merely need to review 
a handful of technical topological conditions on the involved spaces.
As this is a standard step
we defer the checks to \Cref{app:Smale}.

It remains to apply \Cref{lem:homotopy} to $X= V(P)$, $Y = U$ and $f = \pi$.
First, being a union of faces of a hypercube, $U$ is naturally a cubical complex.
In the proof of \Cref{claim:contractible} we have shown that $V(P)$ consists of cubes. These cubes can be subdivided such that they fulfill the conditions of a cubical complex, therefore $V(P)$ is a cubical complex too.

By \Cref{claim:contractible}, it also follows that the pre-image under $\pi$ of every point of $U$ is contractible.
Further, as noted above, $\pi: V(P)\rightarrow U$ is indeed continuous, proper and surjective.
Finally, if $U$ is not connected, we can just consider each connected component separately, as the parts are then also disconnected in V(P).
Thus \Cref{lem:homotopy} shows that $U$ and $V(P)$ are indeed homotopy equivalent.

\section{Tangible Examples}
\label{sec:tangible}
In this section we show that some simple art gallery problem instances have solution spaces homeomorphic to nice topological spaces, in particular, we prove \Cref{thm:tangible}.

\Tangible*

Our examples also work for the \textsc{Point-Vertex Art Gallery} version. Thus, the existence of such examples is not implied by the work of Stade and Tucker-Foltz~\cite{stade2022topological}. Furthermore, our examples are very easy to understand, while the approach of~\cite{stade2022topological} leads to galleries with lots of vertices even for simple spaces like the torus.

Beginning with $1$-dimensional solution spaces will give us insights into the concepts used and starting with \Cref{sub:Spheres} we show how to get higher-dimensional, yet still recognizable solution spaces homeomorphic to commonly known topological spaces.

\subsection{Guard Segments}
\label{sub:guardsegments}

Already in \Cref{sec:Realization} we used pockets in the polygon to ensure that guards stand on certain \emph{guard segments}. There, each guard had its own private segment, and these segments did not interact at all. In particular, the extensions of these segments in both directions (until hitting the border of the polygon) did not intersect. This allowed us to use a very simple construction of two pockets to enforce such a guard segment, see for example \Cref{fig:variables}.

For the following constructions of our example spaces we need guard segments that may intersect. Sadly, our previous construction using two pockets per segment does not enforce guard segments properly in this setting. To see this, note that once a guard gets close to a guard segment from one side, both pockets of the guard segment are guarded if there is another guard close to the segment on the other side. Thus, there are possible solutions in which the polygon is guarded but no guard stands on the segment exactly.

We thus search for a new construction of guard segments which properly enforces the property that in each possible solution, each guard segment needs to contain a guard. We suggest the construction shown in \Cref{fig:guard_segments_new}.

\begin{figure}[ht]
\centering
\includegraphics{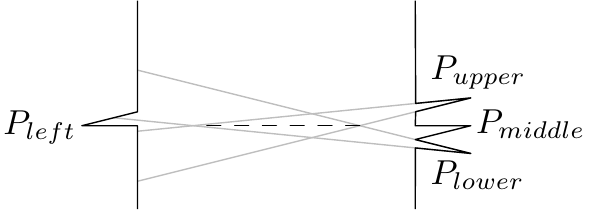}
\caption{Extending the idea of using pockets to enforce a guard segment. Instead of only using one pocket on each side, we use three on one side and one on the other.}
\label{fig:guard_segments_new}
\end{figure}

The idea is that the cones from which these four pockets are visible intersect in exactly the guard segment. Thus, hopefully in an optimal solution any guard would rather be on this line segment and so we indeed get our guard segment. In the remainder of this subsection we will be dealing with the technicalities of these four pockets and prove that they really enforce guard segments.

Without additional arguments, the now four pockets per guard segment only guarantee that at least one guard stands in the union of the four visibility areas of each of the pockets.
This union is some thin area around the extension of the guard segment in both directions. By moving the pockets closer to each other, this area can be made arbitrarily thin. Let us denote this area as the \emph{sausage} around the guard segment.

Our goal is to prove that under certain conditions, the set of solutions using $k$ guards such that each \emph{sausage} contains at least one guard is the same as the set of solutions using $k$ guards such that each \emph{guard segment} contains at least one guard. First, note that each guard segment is contained within its sausage. Thus, it trivially follows that the second set is always a subset of the first. For the other direction, we phrase the following lemma.

\begin{lemma}\label{lem:guard-segments}
In any solution $S$, a guard $g$ standing in the sausage of guard segment $L$ is forced by pockets as shown in \Cref{fig:guard_segments_new} to stand on $L$ itself, if $S$ fulfills either of the following conditions:
\begin{itemize}
    \item[(a)] There is no guard except $g$ standing in the sausage of $L$.
    \item[(b)] There is one other guard $g'\in S$ also in the sausage of $L$, but not exactly on the line extending $L$. There is another guard segment $L'$, such that $L\cap L'\not=\emptyset$, and $L'$ completely lies on the same side of $L$ as $g'$. Furthermore, $g$ is enforced to stand on $L'$ exactly. 
    \item[(c)] There is one other guard $g'\in S$ also in the sausage of $L$, but not exactly on $L$. There are two other guard segments $L_1,L_2$, such that $L_1\cap L$ and $L_2\cap L$ are exactly the two endpoints of $L$. Furthermore, $g$ is enforced to stand on $L_1$, and $g'$ is enforced to stand on $L_2$. 
\end{itemize}
These three cases correspond to the three arrangements shown in \Cref{fig:guard-segment-arrangements}.
\end{lemma}

\begin{figure}[ht]
\centering
	\begin{subfigure}{.32\linewidth}
		\centering
		\includegraphics[page = 2]{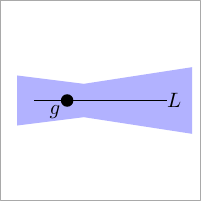}
		\caption*{(a)}
	\begin{minipage}{.1cm}
	\vfill
	\end{minipage}
    \end{subfigure}
    \hfill
	\begin{subfigure}{.32\linewidth}
		\centering
		\includegraphics[page = 3]{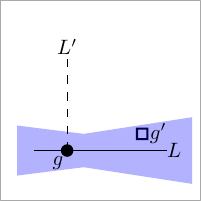}
		\caption*{(b)}
	\begin{minipage}{.1cm}
	\vfill
	\end{minipage}
	\end{subfigure}
	\hfill
	\begin{subfigure}{.32\linewidth}
		\centering
		\includegraphics[page = 4]{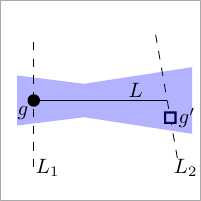}
		\caption*{(c)}
	\begin{minipage}{.1cm}
	\vfill
	\end{minipage}
	\end{subfigure}
	\caption{The three situations describing the fundamental ways how other guards (square) are allowed to enter the sausage such that the guard is still forced to be on the guard segment exactly.}
	\label{fig:guard-segment-arrangements}
\end{figure}

\begin{proof} We prove this separately for each condition, and always assume that the respective condition holds. \begin{itemize}
    \item[(a)] $g$ must be on $L$, as no other guard can see any pocket realizing $L$, and the four pockets can only be seen simultaneously from $L$.
    \item[(b)] Without loss of generality, assume guard $g'$ is above $L$. $g'$ can at most see pockets $P_{left}$ and $P_{lower}$ of $L$. By the condition (b), $g$ must be on guard segment $L'$, but there must also be a guard seeing pockets $P_{upper}$ and $P_{middle}$ of $L$. There exists only a single point on $L'$ which can see both of these, which is the intersection point of $L$ and $L'$, and $g$ must therefore stand at this point.
    \item[(c)] Assume $g'$ is below $L$, and $L\cap L_2$ is the right endpoint of $L$. We can furthermore assume that the angle between $L$ and $L_2$ is large enough that any point on $L_2$ except $L_2\cap L$ can only see either 
    $P_{upper}$ or $P_{lower}$ of $L$, but not both. If this assumption would be wrong, we could further move the pockets of $L$ closer together until it holds. We thus know that $g'$ (which is on $L_2$) can at most see $P_{middle}$ and $P_{lower}$. By the condition, $g$ must be on $L_1$, and we know that there must be some guard seeing pockets $P_{upper}$ and $P_{left}$. Again, assuming the angle between $L$ and $L_1$ not being too small, the only point on $L_1$ which can see all of these is the intersection point of $L$ and $L_1$, and $g$ must therefore stand at this point.
\end{itemize}
The proof for (c) works similarly for $L_1$ being at the other endpoint of $L$, or $g'$ being above $L$.
\end{proof}

Of course, our pocket construction may also properly enforce guard segments for some cases not covered by \Cref{lem:guard-segments}, but the three cases (a)--(c) are enough to prove correctness of all of our constructions for the proof of \Cref{thm:tangible}.

\Cref{lem:guard-segments} will allow us to analyze the solution space of our polygons in terms of guard segments instead of sausages.

\subsection{The Circle}
\label{sub:circle}

As a first application of our new construction of guard segments, we show how we can get the simplest yet non-trivial topological space, namely the circle. Recall that each guard segment requires at least one guard to be placed on that segment. 

\begin{figure}[ht]
\centering
	\begin{subfigure}{.41\linewidth}
		\centering
		\includegraphics[page = 2]{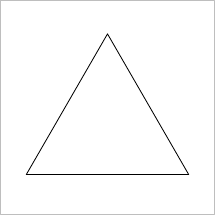}
		\caption{}
		\label{subfiga:Circle}
	\begin{minipage}{.1cm}
	\vfill
	\end{minipage}
    \end{subfigure}
    \hfill
	\begin{subfigure}{.58\linewidth}
		\centering
		\includegraphics[page = 3]{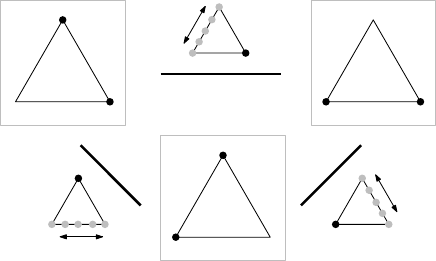}
		\caption{}
		\label{subfigb:Circle}
	\begin{minipage}{.1cm}
	\vfill
	\end{minipage}
	\end{subfigure}
	\caption{If we can force guards onto three specially arranged guard segments forming a triangle (a), then every solution needs two guards. The resulting solution space is homeomorphic to the circle (b).}
	\label{fig:Circle}
\end{figure}

Consider the arrangement of three guard segments given in \Cref{subfiga:Circle}. Note that in this case we need at least two guards to guard everything, as a single guard can guard at most two of the segments simultaneously. Possible solutions are drawn in \Cref{subfigb:Circle}, indicated by black dots. Furthermore, it can be seen that the solution space of a polygon inducing these guard segments is homeomorphic to the circle. Whenever one of the guards is placed at the intersection of two different guard segments, the other guard is free to walk along the remaining unoccupied guard segment. In this way we can start in some configuration, move the first guard, move the second guard, and move the first guard once again and we are back in the configuration we started with. It is important to note that we work with unlabelled guards and thus we really start and end in the exact same configuration. Since there are no solutions with guards not on the triangular set of guard segments one can see that the solution space is indeed homeomorphic to a circle.

\begin{figure}[ht]
\centering
	\begin{subfigure}{.49\linewidth}
		\centering
		\includegraphics[page = 2]{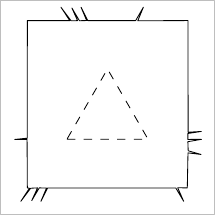}
		\caption{}
		\label{subfiga:Circle_2}
	\begin{minipage}{.1cm}
	\vfill
	\end{minipage}
    \end{subfigure}
    \hfill
	\begin{subfigure}{.49\linewidth}
		\centering
		\includegraphics[page = 3]{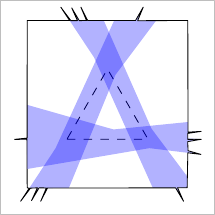}
		\caption{}
		\label{subfigb:Circle_2}
	\begin{minipage}{.1cm}
	\vfill
	\end{minipage}
	\end{subfigure}
	\caption{The second attempt to get guard segments forming a triangle (a). For each guard segment we use four pockets as described in \Cref{fig:guard_segments_new}. In (b) the sausages of the guard segments are indicated in blue.}
	\label{fig:Circle_2}
\end{figure}

We now show that this triangular arrangement of guard segments is enforced properly by our pockets (shown in \Cref{subfiga:Circle_2}). Recalling \Cref{lem:guard-segments}, we only need to prove for each optimal guard placement of guards on the sausages (\Cref{subfigb:Circle_2}), at least one of the three conditions of \Cref{lem:guard-segments} applies. We only show this for the top left solution in \Cref{subfigb:Circle}, for the other solutions the statement follows from symmetry.

For the left and bottom guard segment, there is only one guard on the sausage of the segment, therefore case~(a) applies and the guards are enforced to stand on these segments exactly. Now, for the right segment, for both guards, case~(b) (as well as case~(c)) applies, so at least one of the guards needs to be at an intersection point of segments exactly.

We will refrain from proving the conditions of \Cref{lem:guard-segments} for the more complex constructions, as it is merely tedious but not insightful.
Note that one can always find a sausage containing only a single guard to bootstrap the proof using case (a), as if every sausage would contain two or more guards, a guard could be removed and every sausage would still be guarded, therefore such a guard placement could not be optimal.

\subsection{Clovers}
We now begin proving \Cref{thm:tangible} by showing how to achieve solution spaces homeomorphic to \emph{$k$-clovers}. 
The $k$-clover is also called wedge (sum) of circles or bouquet of circles in the literature.
A $k$-clover consists of $k$ circles glued together at a single point. Note that in the previous section we have shown that we can build a polygon with solution space homeomorphic to the circle, which can also be seen as a $1$-clover.
\paragraph{$2$-clover.}
We first consider the following polygon, see \Cref{subfiga:Clover0}. 

\begin{figure}[ht]
\centering
	\begin{subfigure}{.32\linewidth}
		\centering
		\includegraphics[page = 2]{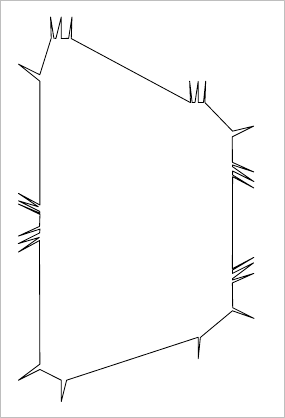}
		\caption{}
		\label{subfiga:Clover0}
	\begin{minipage}{.1cm}
	\vfill
	\end{minipage}
    \end{subfigure}
    \hfill
	\begin{subfigure}{.32\linewidth}
		\centering
		\includegraphics[page = 3]{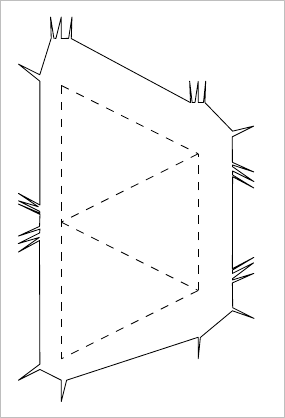}
		\caption{}
		\label{subfigb:Clover0}
	\begin{minipage}{.1cm}
	\vfill
	\end{minipage}
	\end{subfigure}
	\hfill
	\begin{subfigure}{.32\linewidth}
		\centering
		\includegraphics[page = 4]{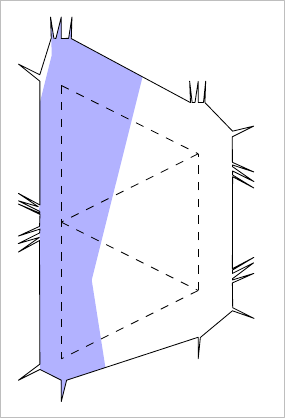}
		\caption{}
		\label{subfigc:Clover0}
	\begin{minipage}{.1cm}
	\vfill
	\end{minipage}
	\end{subfigure}
	\caption{(a) The polygon gives us as a solution space the $2$-clover (i.e., the symbol $\infty$). (b) Guards are forced to live on $6$ specific guard segments (dashed). (c) One of the visibility regions (i.e.,~sausages) is shown on the right, they can be made arbitrarily thin.}
	\label{fig:Clover0}
\end{figure}

The idea is to force an arrangement of $6$ guard segments (cf.~\Cref{subfigb:Clover0}) that each need to have one guard on them. Note that we cannot guard this arrangement with $2$ guards and so any optimal solution needs $3$ guards. As can be seen in \Cref{subfigc:Clover0}, the pockets can be moved together close enough such that the nerve of the sausages of the guard segments is the same as the nerve of the guard segments themselves.

Note that there are multiple solutions with $3$ guards on intersections of guard segments, and for all of these solutions, only one guard can move away from its intersection at a time. \Cref{fig:Clover1} shows a solution with $3$ guards and the movement of one of the guards.

\begin{figure}[ht]
\centering
\includegraphics{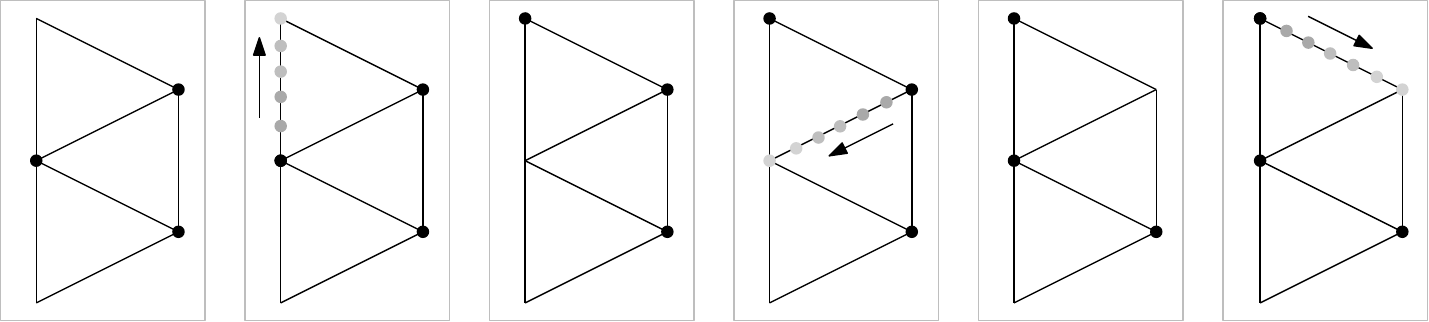}
\caption{The $6$ guard segments from the polygon above. The figure shows how moving guards around is possible, and it shows that there is at least one cycle in the solution space.}
	\label{fig:Clover1}
\end{figure}

By using movements similar to the ones described in \Cref{sub:circle} we can discover the whole solution space, and we arrive at the graph in \Cref{fig:Clover2graph}, where edges mean movement of a single guard along one single guard segment. 

\begin{figure}[htbp]
    \centering
    \includegraphics{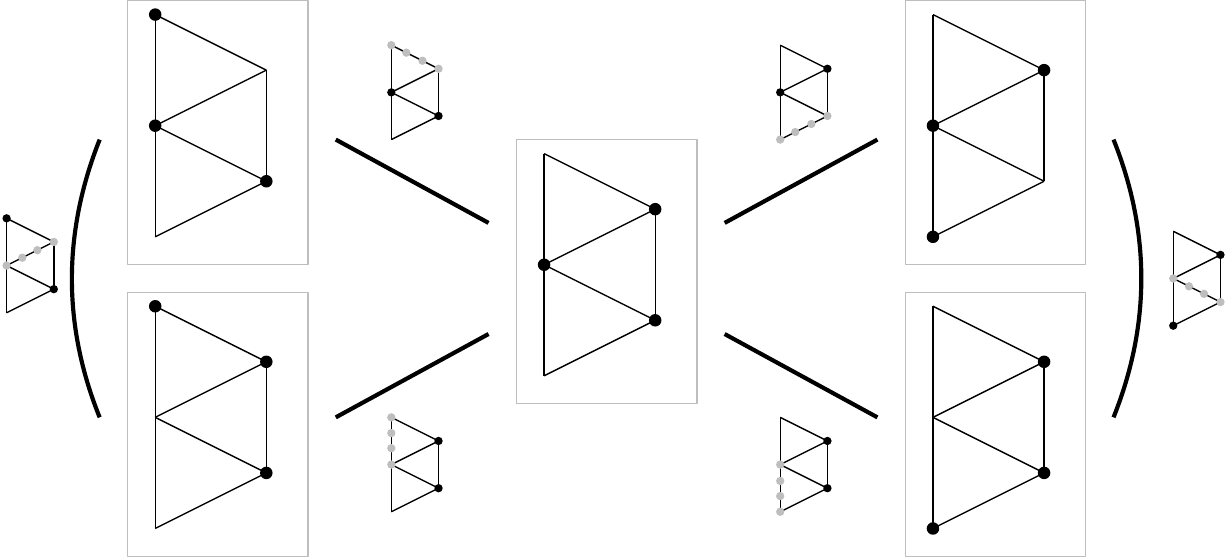}
    \caption{The complete solution space of the polygon. The situations in boxes are connected by single guard movements indicated next to the edges.}
    \label{fig:Clover2graph}
\end{figure}

Finally, using \Cref{lem:guard-segments}, we can conclude that in each solution all guards need to stand on their respective segments, and thus the solution space of this art gallery problem instance is homeomorphic to the symbol $\infty$ or the $2$-clover.

\paragraph{$k$-clovers.}
For showing that we can have any $k$-clover as a solution space we describe the following arrangement of guard segments. Let $k \geq 3$ and create guard segments on the sides of a regular $k$-gon. Additionally, consider the set $M$ of middle points of the already present guard segments and let us add a complete graph onto these (middle) points. 
In other words, we add a guard segment connecting any two points in $M$. Let us denote the resulting set of lines as the \emph{$k$-clover art gallery guard segments}. Examples of the resulting arrangement can be found in \Cref{fig:Clovern}.

\begin{figure}[ht]
    \centering
    \includegraphics{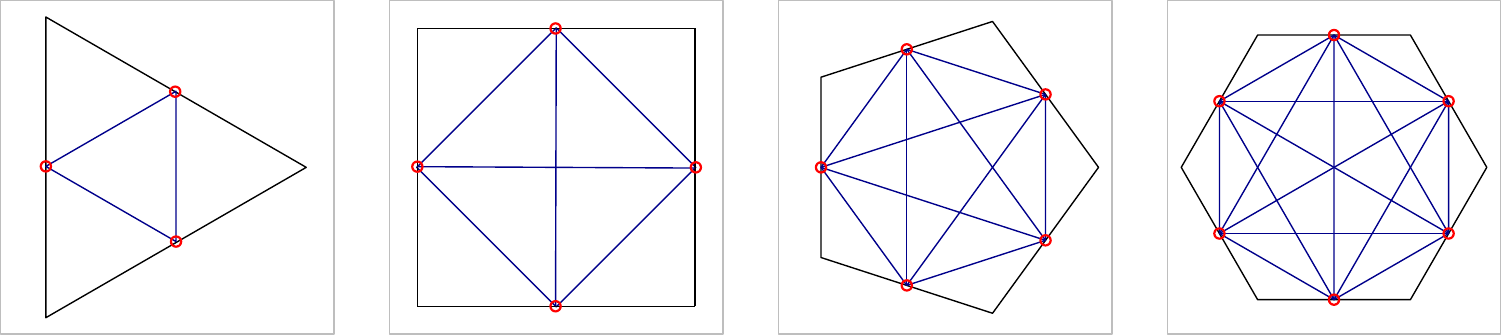}
    \caption{Examples of the $k$-clover art gallery guard segments for small values of $k$. The red circles indicate the set $M$ with the blue lines describing the inner complete graph $K_k$. The black guard segments are the ones we start our construction with, the regular $k$-gon. }
    \label{fig:Clovern}
\end{figure}

For clarity, let us denote the inner part as the inner $K_k$ (blue in \Cref{fig:Clovern}) and the set of intersection points of guard segments as the set of vertices. 

\begin{lemma}
\label{lm:minnumberofguards}
The minimum number of guards to guard every guard segment in the $k$-clover art gallery guard segments is $k$.
\end{lemma}

\begin{proof}
Note that if there is a guard placement guarding everything, then there is a guard placement with the same number of guards, with guards only on vertices. This is because moving a guard from a single guard segment to a vertex only increases the set of segments covered by this guard. 

Assume for the sake of contradiction, that we can guard the inner $K_k$ with at most $k-2$ guards. If all $k-2$ guards are on points in $M$, then there are two unoccupied points in $M$ and the edge connecting them cannot be covered by any guard. Hence, let us assume that $i \geq 1$ of the guards are in the interior of the polygon spanned by the vertices of the inner $K_k$. Consider the induced graph $G$ formed by unoccupied points $M' \subseteq M$. By the construction of the $k$-clover art gallery guard segments, we know that $G = K_{i+2}$ and recall that there is no guard on any point in $M'$. Since every point in $M'$ has $i+1$ adjacent edges in $G$ that do not intersect in the interior of the polygon spanned by the vertices of inner $K_k$, we cannot guard everything with only $i$ guards. Overall this shows that there are at least $k-1$ guards needed in the inner $K_k$. 

Suppose now that we have $k-1$ guards in the inner $K_k$. Then there is at least one unoccupied point of $M$ and with it the corresponding guard segment of the regular $k$-gon is uncovered. Thus any solution to the $k$-clover art gallery guard segments needs at least $k$ guards.

Finally, note that putting a guard on every point of $M$ gives a valid assignment.
\end{proof}

\begin{lemma}
The solution space of the $k$-clover art gallery guard segments is the $k$-clover.
\end{lemma}

\begin{proof}
The concepts used to analyze the solution space of the $2$-clover polygon apply here as well. The common ``point'' (i.e.,~center of the clover) is the canonical configuration already described above, where the $k$ guards are all placed on the points of $M$ (see \Cref{fig:Cloverk} leftmost). Note that in this specific configuration every single guard can be moved along its corresponding outer guard segment.
However; we can only move one of the guards at a time, as moving two of them simultaneously would leave the edge between their original locations unguarded. Note also that moving a guard into the inside of the $K_k$ would leave the outer guard segment (the one from the $k$-gon) unguarded.

Whenever we move a guard far enough, until it reaches a corner of the regular $k$-gon, this frees exactly one other guard (namely the one on the side of the $k$-gon which is newly covered by the moved guard). This freed guard can now only walk on one particular inner edge to the old position of the first guard, ending back at the canonical solution. The described movements can be seen in  \Cref{fig:Cloverk} for the case of the $4$-clover.

\begin{figure}[ht]
\centering
	\begin{subfigure}{.19\linewidth}
		\centering
		\includegraphics[page = 2]{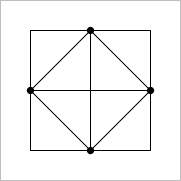}
	\begin{minipage}{.1cm}
	\vfill
	\end{minipage}
        \end{subfigure}
        \hfill
	\begin{subfigure}{.19\linewidth}
		\centering
		\includegraphics[page = 3]{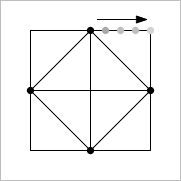}
	\begin{minipage}{.1cm}
	\vfill
	\end{minipage}
        \end{subfigure}
        \hfill
	\begin{subfigure}{.19\linewidth}
		\centering
		\includegraphics[page = 4]{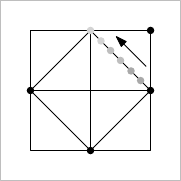}
	\begin{minipage}{.1cm}
	\vfill
	\end{minipage}
	\end{subfigure}
        \hfill
	\begin{subfigure}{.19\linewidth}
		\centering
		\includegraphics[page = 5]{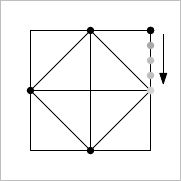}
	\begin{minipage}{.1cm}
	\vfill
	\end{minipage}
	\end{subfigure}
        \hfill
	\begin{subfigure}{.19\linewidth}
		\centering
		\includegraphics[page = 6]{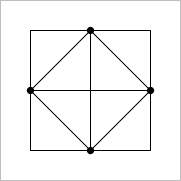}
	\begin{minipage}{.1cm}
	\vfill
	\end{minipage}
	\end{subfigure}
	\caption{Moving guards along the edges of the $k$-clover art gallery guard segments.}
	\label{fig:Cloverk}
\end{figure}

As already mentioned, whenever a guard is on a single edge, no other guard is free to move. Hence, by the structure of the graph, we get exactly $k$ circles attached to a common point and thus the $k$-clover as claimed.
\end{proof}

\begin{lemma}
We can realize the $k$-clover art gallery guard segments as an art gallery problem of a simple polygon. 
\end{lemma}

\begin{proof}
We again use small pockets to force guard segments. One example of a possible resulting polygon for $k=4$ can be found in \Cref{fig:4_Clover_Polygon}. 
\begin{figure}[ht]
    \centering
    \includegraphics{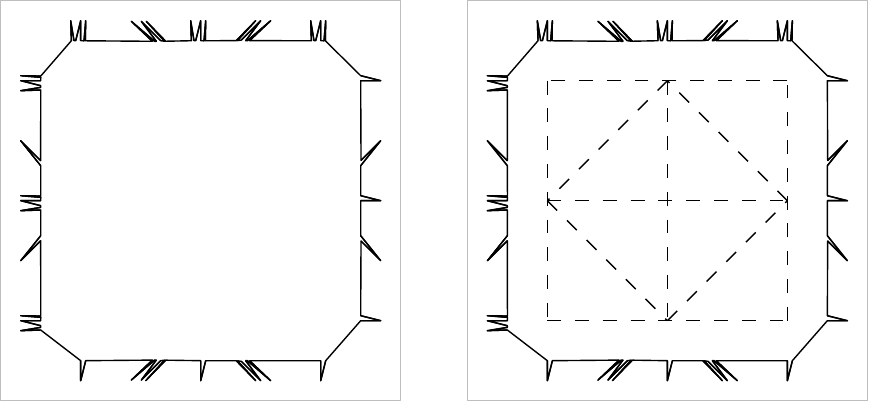}
    \caption{The polygon on the left gives the $4$-clover art gallery guard segments as possible locations for guards. Thus the solution space of this instance of the art gallery problem is homeomorphic to the $4$-clover.}
    \label{fig:4_Clover_Polygon}
\end{figure}
Similarly to the cases before, one can check that every optimal guard placement in this polygon is a solution with guards only on the guard segments, using \Cref{lem:guard-segments}. Thus the polygons describe the $k$-clover art gallery guard segments. 
\end{proof}

\subsection{Chains}
\label{sub:chains}
We continue proving \Cref{thm:tangible}, namely in this section we show that there are instances of the art gallery problem with solution spaces homeomorphic to $k$-chains. The main idea used in this section is to have different sets of guards for the different things we want to chain up. We need to restrict the valid placement of guards with guard segments such that at each moment in time only one set of guards can move freely. 

For simplicity we will restrict ourselves to chaining up circles; however, we believe that the idea can be used to chain up different, more complex spaces as well. Let $k \geq 1$ and consider a path of $k$ vertices. Instead of using vertices we use circles. We denote the resulting set as a \emph{$k$-chain}, some examples for small $k$ can be found in \Cref{fig:k_chain}.

\begin{figure}[tbph]
    \centering
    \includegraphics{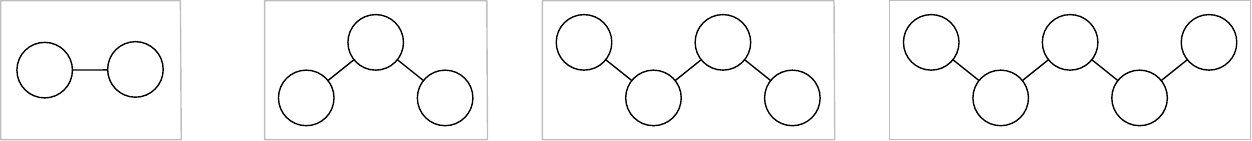}
    \caption{$k$-chains for small values of $k$.}
    \label{fig:k_chain}
\end{figure}

Consider the arrangement of guard segments in \Cref{fig:link_1}. Note that any guard can only cover two guard segments simultaneously and since there are $9$ different segments we need at least $5$ guards. We have already seen in \Cref{sub:circle} that if guard segments form a triangle, we always need two guards on the triangle. Therefore in any optimal guard placement there is a guard on the horizontal guard segment in the middle, we call it the \emph{link guard} on the \emph{link segment}.

\begin{figure}[tbph]
    \centering
    \includegraphics{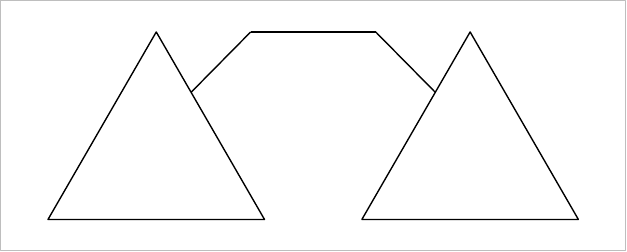}
    \caption{Guard segments forming two triangles that are connected by a link consisting of three individual guard segments.}
    \label{fig:link_1}
\end{figure}

The left and right guard segments of the link ensure that whenever they are not guarded by the link guard, the guards on the corresponding triangle need to maintain one particular position. This means that the link guard can move from one end of its segment to the other in order to lock one circle and unlock the other. In the solution space, this results in a segment linking the two circles, see \Cref{fig:link_2}.

\begin{figure}[ht]
    \centering
    \includegraphics{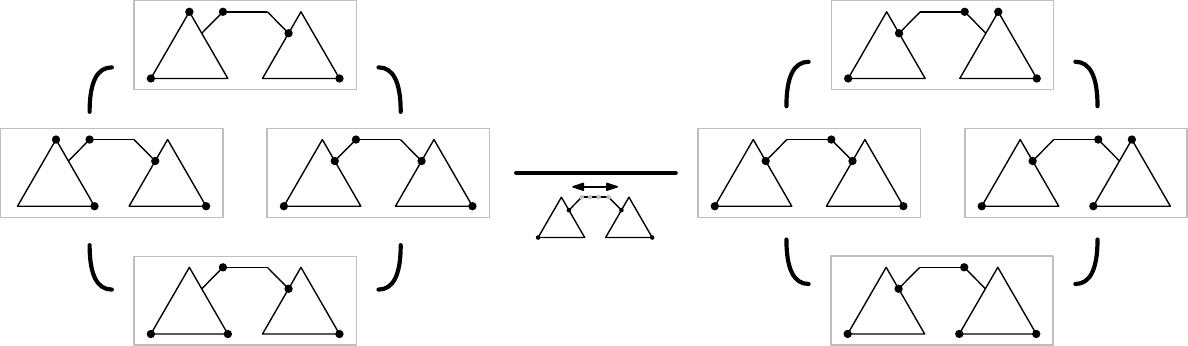}
    \caption{The solution space of the guard segments given in \Cref{fig:link_1} is a $2$-chain.}
    \label{fig:link_2}
\end{figure}

\begin{figure}[ht]
    \centering
    \includegraphics{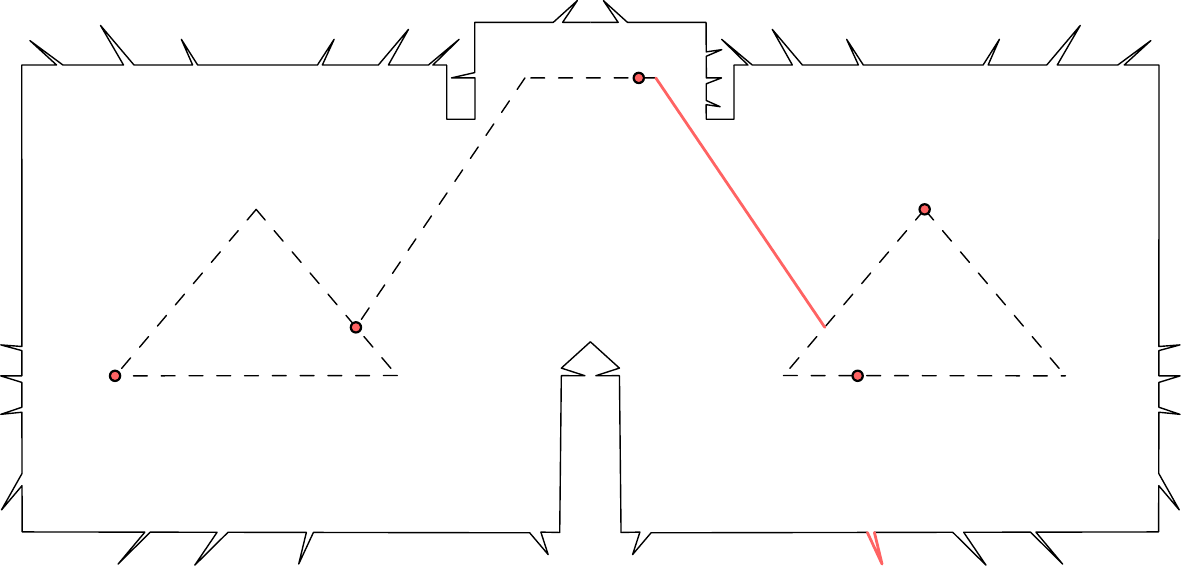}
    \caption{An example of the full polygon forcing guard segments; with solution space homeomorphic to a $2$-chain. Marked is the possible spurious solution that arises from analyzing sausages of guard segments. Note that this placement does not guard the left bottom pocket of the unoccupied guard segment, emphasized in red. This is therefore not actually a spurious solution.}
    \label{fig:linkpoly}
\end{figure}

\Cref{fig:linkpoly} shows one possible instance of the art gallery problem that aims to achieve these desired guard segments. Note that no matter how thin the sausages are made, the sausages of the left and right guard segments of the link actually intersect with the sausages of the bottom triangle segments. We have to argue that this additional intersection does not create spurious solutions that would not assign all guard segments at least one guard. 

Recall that on each sausage there must be one guard. We first observe that the link segment, and each of the triangle segments must contain a guard (by the same arguments as in \Cref{sub:circle}). There is a single type of possible spurious pseudo-solution under these conditions, where all sausages contain a guard, but one of the guard segments of the link does not. This solution is noted in \Cref{fig:linkpoly}, and one can verify that this pseudo-solution can not satisfy the left bottom pocket of the right segment of the link (emphasized in red). We can therefore conclude that in each optimal guard placement, all guard segments must contain a guard.

Note that by repeatedly linking triangles together, we can realize a chain of any length, and the argument above generalizes.

\subsection{Necklaces}
\label{sub:necklace}
We have just described how one can chain up different spaces. In particular we have given art gallery problem instances with solution spaces homeomorphic to chains of circles. Here we want to investigate whether it is possible to close the chain (of circles), that is, we want to construct a necklace of topological spaces, in particular circles. Let us denote a closed $k$-chain as a \emph{$k$-necklace}, see \Cref{fig:k_necklace} for some small examples.
\begin{figure}[h]
    \centering
    \includegraphics{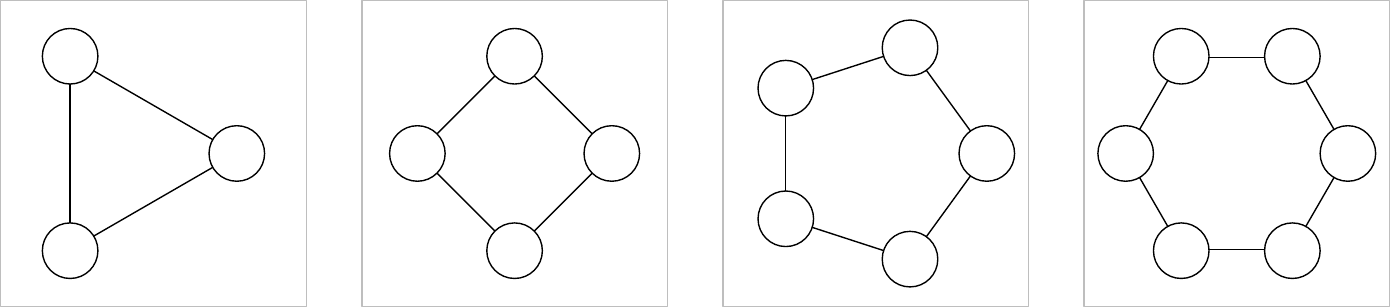}
    \caption{$k$-necklaces for small values of $k$.}
    \label{fig:k_necklace}
\end{figure}

The two topological spaces, chains and necklaces are very similar, the only difference being at the ends of the chain. When we move in the chain, say from the left-most circle to the right-most circle, we want to have some way of going back to the first circle without going through all other circles in-between. This would then close the chain and we would end up with some necklace. 

A first attempt is to link the first and last circle of the chain in a similar way as we linked circles in-between. Then there are as many links as circles. Recall that the link guards cannot guard the whole link, therefore the guards on the circles are forced to guard one of the link segments. Thus, we can never get a solution space homeomorphic to a necklace. So, the construction we present here uses another idea: we duplicate the chain using small guard segment squares and link up the ends of the different copies.
\begin{figure}[ht]
\centering
	\begin{subfigure}{.40\linewidth}
		\centering
		\includegraphics[page = 1]{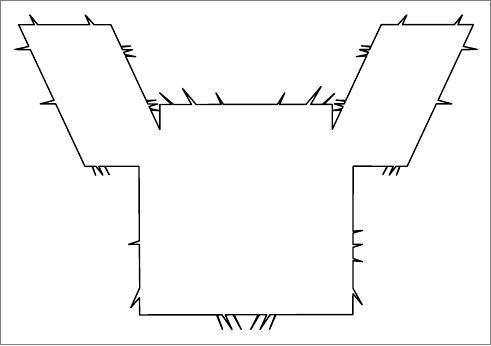}
		\caption{}
		\label{subfiga:4_necklace_polygon}
	\begin{minipage}{.1cm}
	\vfill
	\end{minipage}
    \end{subfigure}
    \hfill
	\begin{subfigure}{.59\linewidth}
		\centering
		\includegraphics[page = 2]{figures/Necklaces/4_necklace_polygon.pdf}
		\caption{}
		\label{subfigb:4_necklace_polygon}
	\begin{minipage}{.1cm}
	\vfill
	\end{minipage}
	\end{subfigure}
	\caption{(a) A polygon and (b) its induced guard segments. They form a simple triangle with two attached flags. The solution space is homeomorphic to the $4$-necklace.}
	\label{fig:4_necklace_polygon}
\end{figure}

Consider the polygon given in \Cref{subfiga:4_necklace_polygon}. The induced guard segments are drawn in \Cref{subfigb:4_necklace_polygon}, where we omitted the polygon for clarity. The guard segments form a $1$-chain (the triangle in the middle) with two small squares attached. Let us denote the small squares, as the \emph{flags}. They are attached at different points of the triangle. Note that if the guards are forced to live on the indicated guard segments, then we need at least $6$ guards in any optimal solution. This is because there are $11$ segments and each guard can occupy at most two simultaneously. Furthermore, two guards need to be on the triangle as well as on each of the flags.

Note that the squares have exactly two discrete solutions and thus by adding the flags to both ends of the triangle will not only duplicate the chain but multiply it by four. In other words, any solution on the triangle can be completed to four different solutions on all guard segments. Note that the flags can only change their state if the guards on the adjacent triangle are already guarding the guard segment that connects the flags to the triangle. This is only possible in exactly one placement of guards. Further the guards in the flags can only move horizontally. Analyzing the possible movements of guards allows us to discover the solution space, \Cref{fig:necklace4} shows that we arrive at a $4$-necklace. 

\begin{figure}[h]
    \centering
    \includegraphics{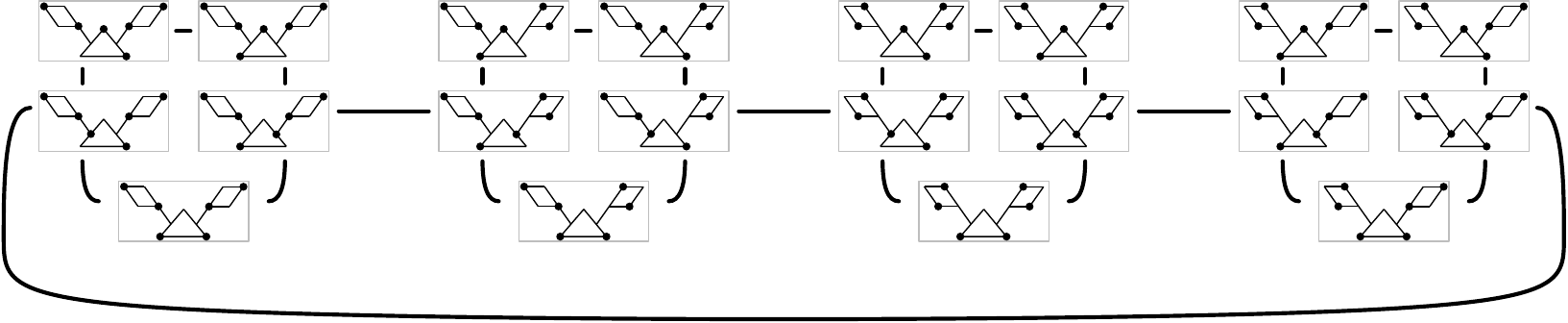}
    \caption{The different guard placements and their connections for guards standing on guard segments as given in \Cref{subfigb:4_necklace_polygon}. The solution space is homeomorphic to a $4$-necklace.}
    \label{fig:necklace4}
\end{figure}

It remains to show that we can indeed force guards onto the guard segments given in \Cref{subfigb:4_necklace_polygon}. By the exact same reasons as presented in \Cref{sub:chains}, it follows that the polygon given in \Cref{subfiga:4_necklace_polygon} does the job.

All the ideas and arguments presented here generalize to any chain. By attaching flags at opposite ends of the chain, we multiply the solution space. Since flags only change their state in one guard placement on the chain, the copies of the chain are linked together nicely. We can therefore get an instance of the art gallery problem with solution space homeomorphic to any $4k$-necklace. We believe that by a slight modification we could also get $2k$-necklaces; however, it remains an open question whether there are simple instances of the art gallery problem with solution spaces homeomorphic to necklaces of odd length.

\subsection{Balls and Spheres}
\label{sub:Spheres}

For the remaining parts of the proof of \Cref{thm:tangible}, we consider multidimensional topological spaces. For simplicity, let us once more study how to get a simple polygon with a solution space homeomorphic to the circle, or sphere $S^1$. This time, the construction will generalize to higher dimensions and allow us to show that there are instances of the art gallery that have solution spaces homeomorphic to balls and spheres. Let us define an \emph{$(h,v)$-art gallery grid}, that is a (graph-theoretical) grid, 
with $h$ horizontal guard segments and $v$ vertical guard segments. Again, let us denote the intersection points of guard segments as \emph{vertices}. Since the vertical and horizontal segments are disjoint, we exactly know how many guards are needed in any optimal solution of an art gallery grid.

\begin{observation}\label{obs:numbgrid}
    On any $(h,v)$-art gallery grid, the minimum number of guards required is $\max(h,v)$. If $h\geq v$, then in any optimal solution there is no guard solely on a vertical segment (and vice versa). 
\end{observation}

\begin{claim}
\label{claim:S1}
The $(3,2)$-art gallery grid has a solution space homeomorphic to $S^1$.
\end{claim}

\begin{proof}
By \Cref{obs:numbgrid}, the $3$ guards are always on horizontal guard segments. Note that whenever one of the guards is not on a vertex, then the other two guards are fixed in one of two (antipodal) states. Similarly, when all three guards are on vertices, we can move two of them, but only one simultaneously. There is a way to go through all the states with guards only on vertices using all the different states and taking no detours, i.e.,~a Hamiltonian cycle. Therefore the solution space of the $(3,2)$-art gallery grid is homeomorphic to $S^1$. For some intuition, consider \Cref{fig:1sphere} below. 
\begin{figure}[ht]
\centering
	\begin{subfigure}{.60\linewidth}
		\centering
		\includegraphics[page = 2]{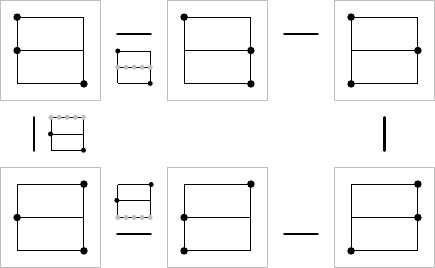}
		\caption*{(a)}
	\begin{minipage}{.1cm}
	\vfill
	\end{minipage}
    \end{subfigure}
    \hfill
	\begin{subfigure}{.39\linewidth}
		\centering
		\includegraphics[page = 3]{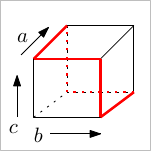}
		\caption*{(b)}
	\begin{minipage}{.1cm}
	\vfill
	\end{minipage}
	\end{subfigure}
	\caption{The solution space of the $(3,2)$-art gallery grid drawn in two different ways. On the left one can see what happens in the grid. 
    On the right we embedded the grid into the cube, one dimension for each guard. The solution space is homeomorphic to the circle $S^1$.}
	\label{fig:1sphere}
\end{figure}
\end{proof}

More generally, we can state the following result, namely the part of \Cref{thm:tangible} about spheres. 
\begin{lemma}
\label{lem:Sk}
For any $k \geq 1$, the solution space of the $(k+2,2)$-art gallery grid is homeomorphic to $S^k$.
\end{lemma}

Before we are able to prove this statement we take a closer look at the $(k+1,1)$-art gallery grids that arise if we ignore one of the vertical lines in the $(k+1,2)$-art gallery grids, see \Cref{fig:balls} for an example. On the one hand they give interesting solution spaces and on the other hand these grids are an essential part in proving  \Cref{lem:Sk}.

\begin{figure}[ht]
    \centering
    \includegraphics{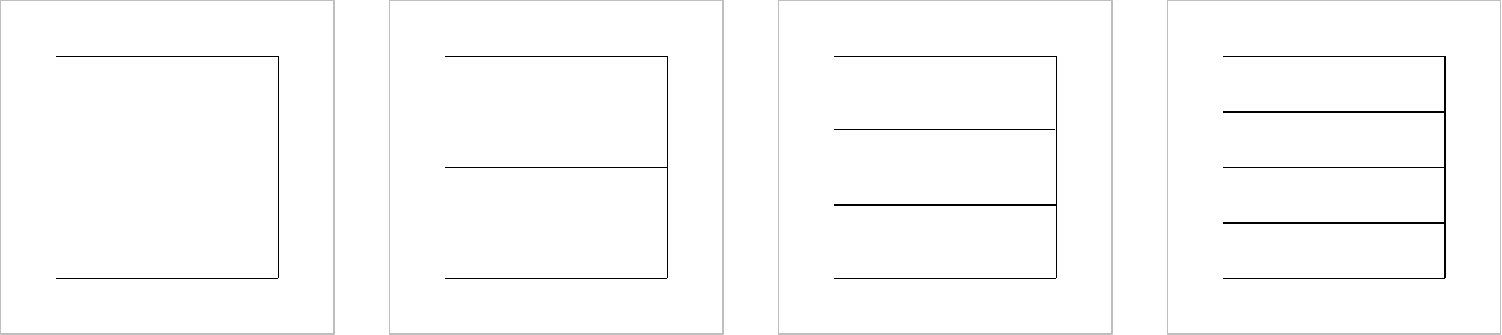}
    \caption{Some $(k+1,1)$-art gallery grids for small values of $k$.}
    \label{fig:balls}
\end{figure}

\begin{lemma}
\label{lem:balls}
For any $k \geq 2$, if the $(k+1,2)$-art gallery grid has a solution space homeomorphic to~$S^{k-1}$, 
then the $(k+1,1)$-art gallery grid has a solution space homeomorphic to the ball $B^k$.
\end{lemma}

\begin{proof}
Let us consider the solution space of the $(k+1,1)$-art gallery grid.
Among all solutions, there is one special one, that is the guard placement with all $k+1$ guards on the vertical line. Let us denote this placement as the \emph{$1$-placement}. This solution will be the center of our ball. We can now embed the solutions of the $(k+1,2)$- and $(k+1,1)$-art gallery grids into the $k+1$-dimensional hypercube (one dimension for each guard along its horizontal guard segment). 

Look at all the solutions of the $(k+1,1)$-art gallery grid that have at least one guard on the missing left vertical guard segment, and at least one guard on the remaining right vertical segment. 
These \emph{boundary solutions} are also solutions of the $(k+1,2)$-art gallery grid, and form an $S^{k-1}$ by assumption. 
Every solution $s$ of the $(k+1,1)$-art gallery grid lies on a straight line between such a boundary solution and the $1$-placement~---~the boundary solution can be found by continuing the ray from the $1$-placement through $s$ until the first guard reaches the left endpoint of its segment. Furthermore, every placement $p$ on such a line between a boundary solution and the $1$-placement must be a valid solution of the $(k+1,1)$-art gallery grid, as some guard is always on the vertical segment for all placements on this line. We conclude that the solution space of the $(k+1,1)$-art gallery grid is homeomorphic to the filled $S^{k-1}$, i.e., the ball $B^k$.
\end{proof}

Now we are able to prove \Cref{lem:Sk}. 

\begin{proof}[Proof of \Cref{lem:Sk}]
In any optimal solution there are exactly $k+2$ guards, one on each horizontal guard segment (cf.~\Cref{obs:numbgrid}). We give a proof by induction over $k$, where the base case $k=1$ is proven already in \Cref{claim:S1}. Therefore let us assume that $k \geq 2$. 

Fix any of the guards, and let us denote it as guard $a$. If $a$ has to be free to move around, then the remaining $k+1$ guards need to cover both vertical guard segments. They are arranged in a grid-structure; thus, we are in the situation of the $(k+1,2)$-art gallery grid. By the induction hypothesis we already know that the corresponding solution space is homeomorphic to $S^{k-1}$. Note that moving $a$ along its guard segment multiplies the $S^{k-1}$ by an interval.
 
On the other hand, if $a$ lies on one of the vertical lines, then the remaining guards only need to cover the other vertical guard segment. Hence, we are in the situation of the $(k+1,1)$-art gallery grid. By the induction hypothesis, and by \Cref{lem:balls}, we know that its solution space is homeomorphic to $B^k$.

For both the middle hypercylinder (for $a$ free), and the two balls (for $a$ on a vertical guard segment), the boundary is exactly the set of those solutions in which $a$ is free to move, but on one of the vertical guard segments. These are also exactly the solutions through which one can move from a ball to the hypercylinder and vice-versa, and we can conclude that the parts therefore fit together. We see that the solution space is homeomorphic to two balls $B^k$ joined by a hypercylinder $S^{k-1}\times [0,1]$, which is homeomorphic to $S^k$ as claimed.
\end{proof}

\begin{figure}[ht]
    \centering
    \includegraphics{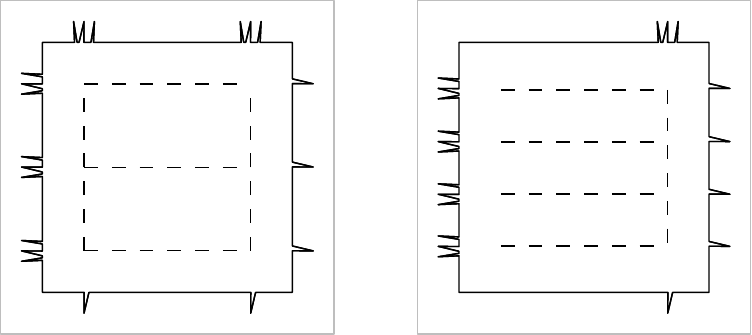}
    \caption{Polygons realizing the $(3,2)$-art gallery grid and the $(4,1)$-art gallery grid, respectively.}
    \label{fig:grids_polygon}
\end{figure}

Showing the realizability of any grid is not hard using \Cref{lem:guard-segments}. \Cref{fig:grids_polygon} gives polygons for the $(3,2)$-art gallery grid and for the $(4,1)$-art gallery grid. As a final observation in this section, we shortly take a closer look at one of the smallest interesting grids and analyze the $(2,2)$-art gallery grid, see \Cref{fig:22grid}. The following observation can be seen as an extension of \Cref{lem:Sk}.

\begin{observation}\label{obs:22grid}
    The $(2,2)$-art gallery grid has a solution space with two disjoint points, that is, homeomorphic to $S^0$.
\end{observation}

\begin{figure}[ht]
    \centering
    \includegraphics{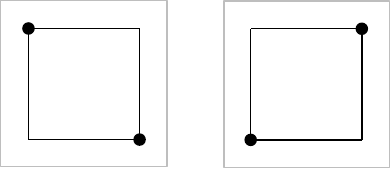}
    \caption{Solution space of the (2,2)-art gallery grid.}
    \label{fig:22grid}
\end{figure}

\subsection{Torus}
\label{sub:Torus}

We start this section with the following observation about independent guard segments in a polygon. If a polygon can be split into two parts, such that the guards guarding each part can move around independently from each other (i.e.,~the possible positions of the guards guarding one part do not depend on the positions of the guards guarding the other part), then the solution space of the whole polygon is homeomorphic to the Cartesian product of the two solution spaces given by the possible positions of the guards guarding each part on its own.

With this in mind, we can now construct an art gallery problem instance with a solution space homeomorphic to a torus. Note that the torus is homeomorphic to the Cartesian product of two circles, so we will use the circle construction from \Cref{sub:circle}. Now all we need to do is to fit two copies of this construction into the same polygon, with guards from each circle construction not influencing the others. The key idea is to separate the two circle constructions in the polygon such that a guard from one side cannot see any of the pockets of a guard segment from the other set. This can be achieved by the construction in \Cref{fig:torus}. The middle separator guarantees that the sausages of any guard segment on the left and any guard segment on the right do not intersect.

\begin{figure}[ht]
    \centering
    \includegraphics{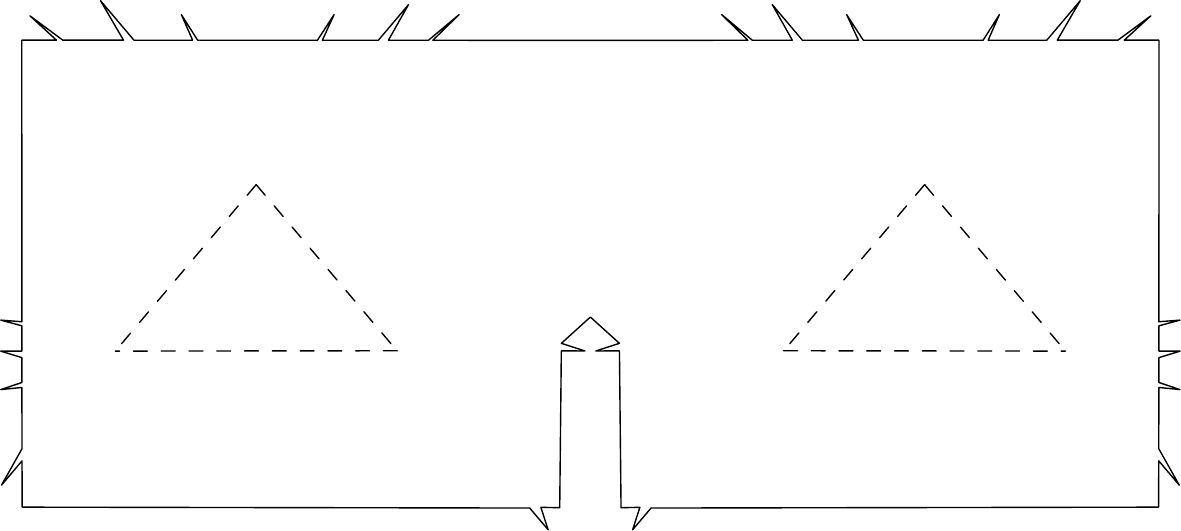}
    \caption{The polygon with solution space homeomorphic to the torus. Guards are forced onto guard segments forming two independent triangles.}
    \label{fig:torus}
\end{figure}

\subsection{Double Torus}
\label{sub:doubleTorus}
\begin{figure}[ht]
\centering
	\begin{subfigure}{.49\linewidth}
		\centering
		\includegraphics[page = 2]{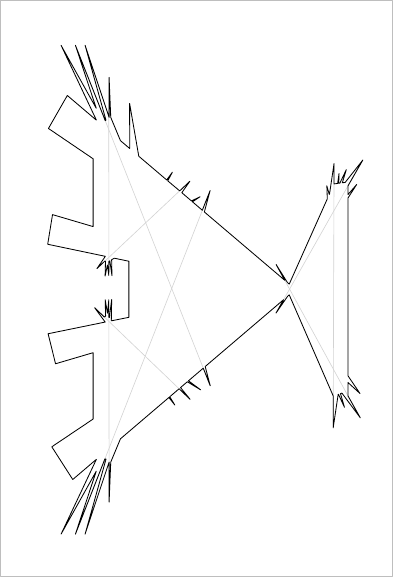}
		\caption{}
		\label{subfiga:doubletorusfull}
	\begin{minipage}{.1cm}
	\vfill
	\end{minipage}
    \end{subfigure}
    \hfill
	\begin{subfigure}{.49\linewidth}
		\centering
		\includegraphics[page = 3]{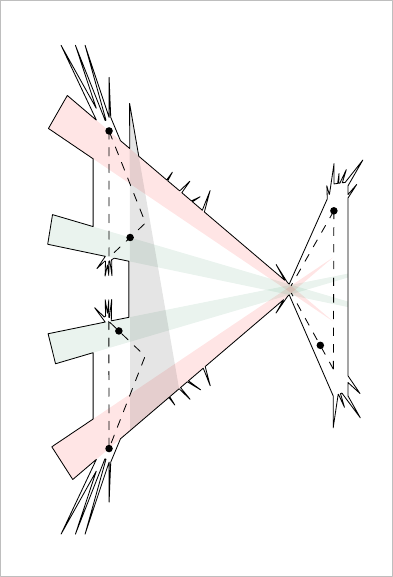}
		\caption{}
		\label{subfigb:doubletorusfull}
	\begin{minipage}{.1cm}
	\vfill
	\end{minipage}
	\end{subfigure}
	\caption{The polygon with solution space homeomorphic to the double torus. Guards are forced to live on guard segments, forming two guard triangles and an open triangle (dashed). An optimal solution needs only six guards. There are additional restrictions on the exact placements of the guards by \emph{circle pockets} (visibility areas in green and red) and a \emph{hole pocket} (visibility area in gray). We split the analysis of this whole polygon into three different components.}
	\label{fig:doubletorusfull}
\end{figure}

As a last tangible example we give a polygon with solution space homeomorphic to a double torus. The polygon has many spikes and pockets (see \Cref{fig:doubletorusfull}) but the key idea is once more to force guards onto guard segments. In particular the guard segments form two triangles and a path of three segments that nearly forms a triangle. Note that this nearly-triangle is not a triangle as it is not closed at its left. There are additional pockets on the left hand side and on the top part of the polygon that restrict the possibilities to place guards further. Let us denote the four pockets on the left \emph{circle pockets}, whose visibility areas are indicated in green and red. Let us call the additional pocket on the top part of the polygon the \emph{hole pocket}. The visibility area of the hole pocket is indicated in gray. Note that all other pockets serve the purpose of forcing guards onto the segments, as in previous sections. 

It is worth mentioning two features of the polygon.
First, we consider the intersections of the visibility areas of the hole pocket and each circle pocket. These intersections all contain only a single point lying on guard segments. Second, from both endpoints of the guard segment path on the right side, all circle pockets can be seen.

To simplify the analysis, we mainly consider the following three components: (1) the guard segment triangle on the top left, (2) the guard segment triangle on the bottom left, and (3) the three guard segments on the right, later referred to as the top, bottom, and vertical segment of the third component. We will prove that each of these components will need to be guarded independently, in other words no guard from one component can see the pockets of the guard segments from another component. Apart from these components the circle pockets and the hole pocket put additional constraints on the possible guard placements and create dependencies across components. In the following we will focus on the different components of the construction and show how they interact. Our goal is proving the following lemma.

\begin{lemma} \label{lem:doubletorus}
The solution space of the polygon given in \Cref{fig:doubletorusfull} is homeomorphic to a double torus.
\end{lemma}

We start by proving the following lemma. It shows that in optimal solutions, guards are indeed forced to live on guard segments. In particular the guards behave as expected and there have to be two guards in each component. The proof only relies on the way the polygon is drawn and does not give additional insights.

\begin{lemma} \label{lem:doubletorus6}
The polygon in \Cref{fig:doubletorusfull} needs 6 guards to be guarded. In an optimal solution with 6 guards, the guards fulfill the following conditions:
\begin{itemize}
    \item[(a)] there are two guards on segments of component (1), and two on segments of component (2).
    \item[(b)] there is one guard on the top segment and one on the bottom segment of component (3).
\end{itemize} 
\end{lemma}

\begin{proof}
To see that 6 guards are sufficient, it suffices to give one configuration with 6 guards seeing the entire polygon. Observe that the guard placement in \Cref{fig:doubletorusfull} is a valid solution. Next we prove that 6 guards are necessary.

Note that component (3) is completely isolated from the others. Furthermore, there is no common intersection of the visibility areas of all pockets of component (3), and thus two guards are required to guard the component. Where can those guards be? Observe that there must be a guard in the sausage of the vertical segment, denote it as guard $g$ and let us denote the other guard as $h$. Suppose that $g$ is not in the sausage of any of the other segments of the component. In this case $h$ must guard all the pockets of both the top and bottom guard segments, which is not possible, as the segments do not intersect. We conclude that guard $g$ must be in one other sausage as well and guard $h$ must guard the third sausage. Using \Cref{lem:guard-segments}, we can show that the guards must stand on guard segments exactly, proving property (b).

It remains to show that in any optimal solution, four more guards are required; two on each of the triangles. In \Cref{sub:Torus} we have shown that two separate triangular arrangements of guard segments require four guards exactly on the guard segments. The only difference of that construction to the polygon here is that we could potentially have guards guarding pockets from both triangles simultaneously. The only possibility for that to happen is if at least one guard stands in the intersection of the sausages of the two \emph{offending segments} of the triangles: those which point to the middle of the polygon. Let us denote this intersection of sausages as the \emph{middle}. Note that no guard in the middle can see pockets of any non-offending guard segment, thus we need at least one guard per triangle that can only see pockets of the corresponding component. Consequently there are at most two guards in the middle in any optimal solution.

Any guard in the middle can see at most two pockets of each guard triangle, see \Cref{fig:double_torus_middle}. Therefore if there is only one guard in the middle we require at least two additional guards in components (1) and (2) each, leading to a total of seven guards. Thus, having only one guard in the middle is not optimal. 

\begin{figure}[ht]
\centering
	\begin{subfigure}{.54\linewidth}
		\centering
		\includegraphics[page = 2]{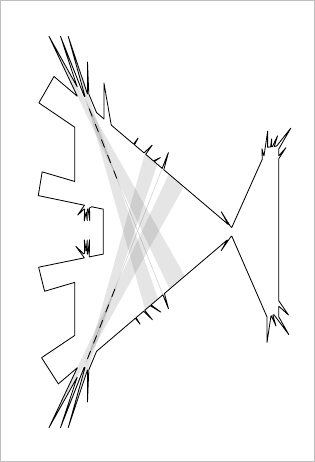}
		\caption{}
		\label{subfiga:double_torus_middle}
	\begin{minipage}{.1cm}
	\vfill
	\end{minipage}
    \end{subfigure}
    \hfill
	\begin{subfigure}{.45\linewidth}
		\centering
		\includegraphics[page = 3]{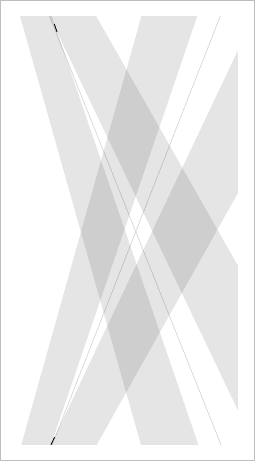}
		\caption{}
		\label{subfigb:double_torus_middle}
	\begin{minipage}{.1cm}
	\vfill
	\end{minipage}
	\end{subfigure}
	\caption{(a) The polygon and (b) the middle with the offending guard segments (dashed), the extension lines of the guard segments (gray lines) and the visibility areas of their respective guard segment pockets (gray). There is only a single point seeing four of the involved pockets, the intersection of the gray lines.}
	\label{fig:double_torus_middle}
\end{figure}

If there are two guards in the middle, one can verify that these two guards cannot guard the two offending segments completely, see \Cref{fig:double_torus_middle}. In particular, there is only a single point that can see four of the eight pockets of the offending segments (the intersection of the extensions of the guard segments). Consequently any two guards in the middle together can see at most six pockets. Hence, in one of the components (1) or (2) more than 8 unguarded pockets remain and two guards are needed, leading to at least seven guards again.

We conclude that there is no optimal solution in which a guard sees a pocket of component (1) and one of component (2) simultaneously. Therefore the two components are independent and as in \Cref{sub:Torus}, there can only be guards on the guard segments, proving property (a).
\end{proof}

Recall that the red and green areas (e.g.~in \Cref{fig:doubletorusfull}) show the visibility areas of the circle pockets on the left side of the polygon. If one of the guards of component (3) is at the left endpoint of its segment, it guards all four circle pockets. The left and right parts of \Cref{fig:doubletorus2states} show the two configurations of the guards in component (3) for which this holds. We refer to them as \emph{unlocking configurations}, and the rest of the configurations of the guards in component (3) are referred to as \emph{locking configurations}. In the following, we will analyze the solution spaces of the first two components, when knowing that the third one is in either the unlocking or the locking configuration. In the end it only remains to show that all parts of the solution space come together nicely and form a double torus.

\begin{figure}[htbp]
    \centering
    \includegraphics{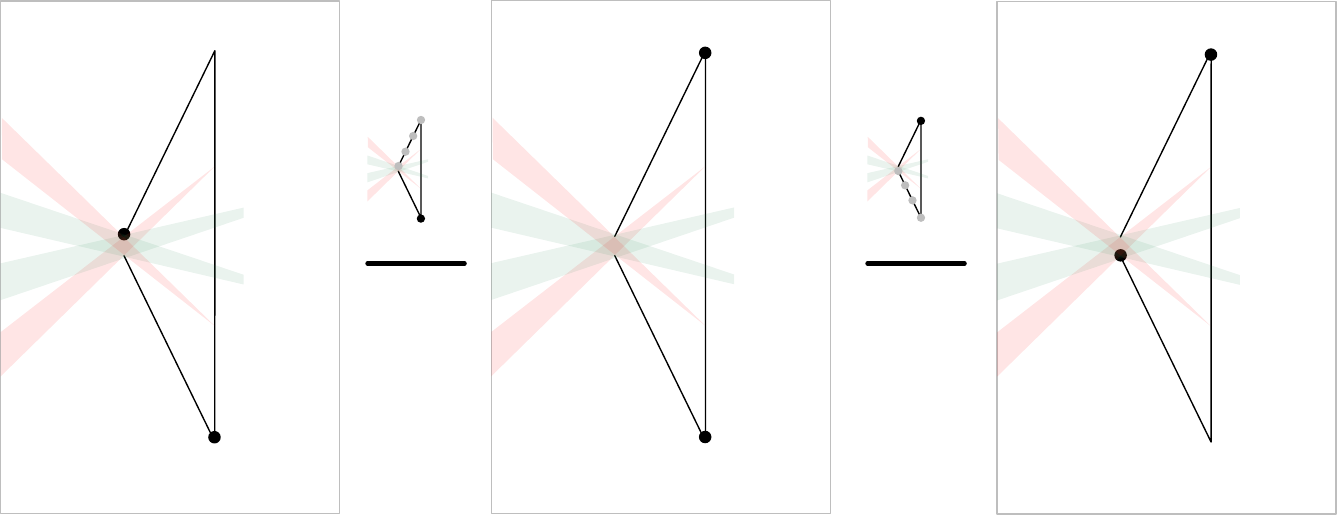}
    \caption{The solution space restricted to the guards of component (3) is homeomorphic to an interval. The left and right figures are the unlocking configurations, i.e.,~all circle pockets are already guarded. All configurations in between are locking.}
    \label{fig:doubletorus2states}
\end{figure}

\begin{lemma}\label{lem:torushole}
    The solution space, restricted to solutions in which the guards of component (3) are in one of the unlocking configurations, is homeomorphic to a torus with a hole.
\end{lemma}

\begin{proof}
Since a guard from component (3) can see the circle pockets, they do not affect the position of the rest of the guards, so we can ignore these pockets. Additionally, we can also completely ignore the placement of guards in component (3) since its guards cannot see any pockets from the other components (cf.~\Cref{lem:doubletorus6}), nor can they see the hole pocket. This leaves us with components (1) and (2), which on their own would have a solution space homeomorphic to a torus, ignoring the hole pocket. To prove the lemma, we need to argue that the hole pocket eliminates a part of the solution space (we call it the \emph{forbidden region}) that is homeomorphic to a disk, resulting in a torus with a hole. 

Observe that the forbidden region is made of all placements of the guards in components (1) and (2), for which the hole pocket is not guarded. The guard placements in the two components are independent, since the hole pocket constraint applies to both. Following the reasoning in \Cref{sub:Torus} it suffices to prove that the forbidden region of the guards of each component is homeomorphic to a line segment.

In each component we have guard segments forming a triangle, and some area around one of the vertices from which the hole pocket can be seen. \Cref{fig:torusholeforbid} shows the guard placements that do not guard the hole pocket. As this is homeomorphic to a line segment, the forbidden region is homeomorphic to a $B^2$, and the lemma follows.
\begin{figure}[htbp]
    \centering
    \includegraphics{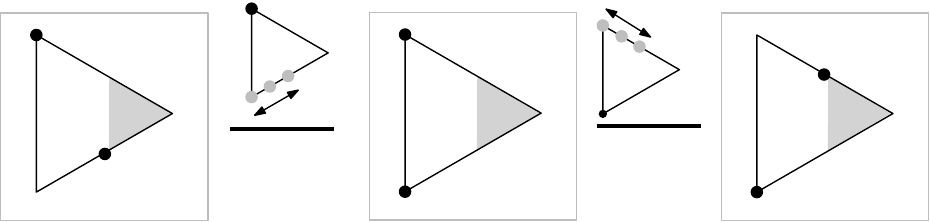}
    \caption{The forbidden region in the solution space of one of the components (1) or (2). The hole pocket is only visible from the gray area; therefore, if either of the guards is in the gray area, then the hole pocket is visible and the solution is not forbidden.}
    \label{fig:torusholeforbid}
\end{figure}
\end{proof}

\begin{lemma}\label{lem:toruscylinder}
    The solution space, restricted to solutions in which the guards of component (3) are in some locking configuration, is homeomorphic to a circle $S^1$.
\end{lemma}

\begin{proof}
Recall that by \Cref{lem:doubletorus6} no guard can be in the interior of the vertical guard segment. Therefore in any locking configuration, the guards of component (3) do not see the circle pockets, the hole pocket or the pockets of components (1) and (2). This also means that the placement of guards in component (3) is independent from the rest of the polygon; the solution space for this component is homeomorphic to a line segment, with the unlocking configurations at the ends as boundary (see \Cref{fig:doubletorus2states}).

Considering components (1) and (2), we note that each guard can only see one of the circle pockets at any time, so the four guards each need to be guarding one circle pocket. Note that there is also the hole pocket. Let $I$ be the set of points on a guard segment triangle from which a guard is able to see both an circle pocket and the hole pocket. Observe that there are only four points in $I$ and that at least one of the guards needs to be on such a point. 

Furthermore, as we have seen before, in each component (1) and (2) at least one guard needs to be on a vertex of the triangle. Knowing this, we can see that each guard can only move between a fixed vertex and a point in $I$, when another guard satisfies the corresponding requirement. \Cref{fig:doubletorusstates} shows these possible movements. 
\begin{figure}[ht]
    \centering
    \includegraphics{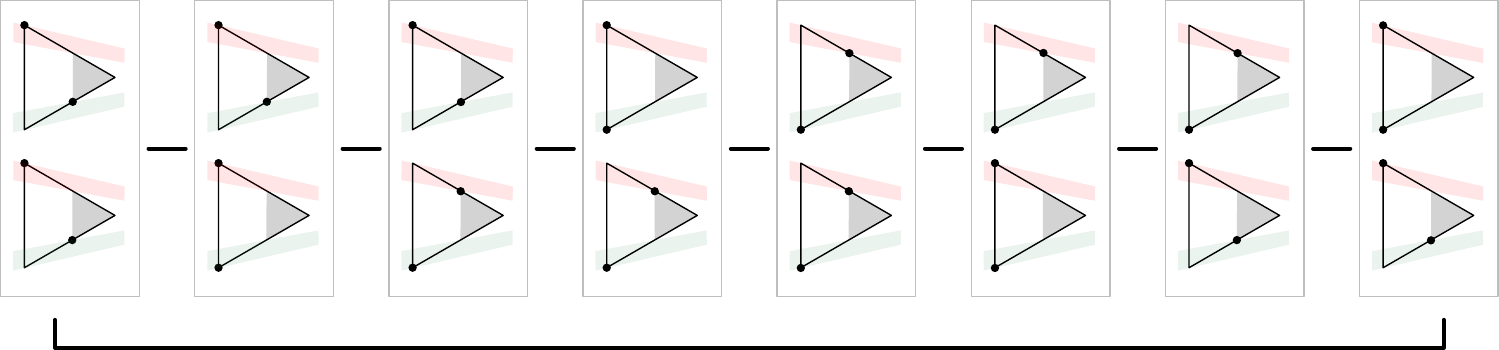}
    \caption{Solution space formed by the guards of components (1) and (2) when the guards of component (3) are in some fixed locking configuration. One of the guards has to see the hole pocket, its visibility area is indicated in gray. There have to be guards in the visibility areas of the circle pockets as well.}
    \label{fig:doubletorusstates}
\end{figure}
One can see that the solution space is homeomorphic to a circle and the desired statement follows. 
\end{proof}

We are finally ready to prove \Cref{lem:doubletorus}, that indeed this art gallery problem instance has a solution space homeomorphic to the double torus.

\begin{proof}[Proof of \Cref{lem:doubletorus}]
We know that in any optimal solution there are exactly six guards, two in each component. Let us assume that the two guards in component (3) are in a locking configuration. Then by \Cref{lem:toruscylinder} we know that that the solution space of the remaining guards is homeomorphic to a circle $S^1$. Since all locking configurations are homeomorphic to an interval, the solution space over all locking configurations is a cylinder (i.e.,~$S^1 \times [0,1]$). On the other hand, if the guards in component (3) are in either unlocking configuration, then the solution space is homeomorphic to a torus with a hole, cf.~\Cref{lem:torushole}.

For both the middle cylinder (component (3) locking) and the two tori with holes (component (3) in either unlocking configuration) the boundary is exactly the set of solutions with guards of components (1) and (2) guarding one of the circle pockets each and guards of component (3) in unlocking configuration. These are also exactly the solutions through which one can move from a torus to the cylinder and vice-versa, and we can conclude that the parts fit together. We see that the solution space is homeomorphic to two tori joined by a cylinder $S^1 \times [0,1]$, which is homeomorphic to the double torus as claimed.
\end{proof}

\section{Conclusion}
\label{sec:conclusion}
We showed that up to homotopy equivalence, any compact semi-algebraic set can be achieved as a solution space of an art gallery problem instance, even when considering the NP-contained version \textsc{Point-Vertex Art Gallery}. By doing this, we not only expand on a result by Abrahamsen, Adamaszek and Miltzow~\cite{ARTETR}, but also highly improve the simplicity of the arguments required. 
We are now left to wonder whether even homeomorphism-universality holds for this simpler version of the art gallery problem.

Furthermore, we have given example instances of the art gallery problem, that are not only homotopy equivalent, but even homeomorphic to well-known topological spaces. In particular, we have shown that we can get $1$-dimensional topological spaces with any number of circles (clovers and chains), and that we can get simple higher-dimensional spaces (spheres). 
On the one hand, a natural open question asks whether it is possible to extend our ideas and find nice examples leading to orientable surfaces of higher genus.
On the other hand, it also remains open whether there are simple art galleries leading to unorientable surfaces like the M\"{o}bius strip, Klein bottle, or the projective plane. To answer this, it could very well be that completely new ideas are needed.

\vfill

\paragraph{Acknowledgments.}
This research started at the 18th Gremo's Workshop on Open Problems (GWOP) in Morschach, Switzerland, 2021.
We thank the organizers for providing a very pleasant and inspiring working atmosphere.
Tillmann Miltzow is generously supported by the Netherlands Organisation for Scientific Research (NWO) under project no. 016.Veni.192.250.
Patrick Schnider has received funding from the European Research Council under the European Unions Seventh Framework Programme ERC Grant agreement ERC StG 716424 - CASe.
We thank anonymous reviewers for their useful feedback.

\newpage
\bibliographystyle{plain}
\bibliography{bibliography/references,bibliography/ETR, bibliography/ArtGallery}

\newpage
\appendix
\section{Proof of \Cref{lem:homotopy}}
\label{app:Smale}

We use the following result of Smale~\cite{smale1957vietoris}.
We will repeat all definitions below.

\begin{theorem}[\cite{smale1957vietoris}, ``Main Theorem'']
Let $X$ and $Y$ be connected, locally compact, separable metric spaces.
Assume further that $X$ is locally contractible.
Let $f:X\rightarrow Y$ be a continuous proper surjective map such that for each $y\in Y$ the space $f^{-1}(y)$ is contractible and locally contractible.
Then $f$ is a weak homotopy equivalence.
\end{theorem}

In the following, we will briefly recall all the notions used in the above statement.
In our setting, where $X$ and $Y$ are cubical complexes, it turns out that all the technical preconditions follow rather easily from embedding the cubical complex into $\R^n$ for some $n$.
Further, weak homotopy equivalence implies homotopy equivalence for a large class of spaces, including cubical complexes.

Let us now dive into the details.

\begin{enumerate}
    \item By the assumption of \Cref{lem:homotopy}, $X$ and $Y$ are connected, $f$ is continuous, proper and surjective, and $f^{-1}(y)$ is contractible.
    \item A space $X$ is called \textit{locally compact} if for each $x\in X$ there exists an open set $V\subset X$ and a compact set $K \subset X$ such that $x\in V\subset K$.
    Any hypercube in Euclidean space $\R^n$ is locally compact, and thus so are cubical complexes, as each cubical complex can be embedded into a hypercube of large enough dimension.
    \item A space $X$ is called \textit{separable} if there exists a countable sequence of elements $\{x_i\}_{i=1}^{\infty}$ such that each non-empty open set contains an element of the sequence.
    The Euclidean space $\R^n$ is separable ($\mathbb{Q}^n$ is an example of the desired sequence), and so are all its subspaces, hence cubical complexes are separable.
    \item A topological space is \textit{metric}, if its topology is given by a metric. Any subspace of $\R^n$ is metric. 
    In fact, for the spaces we are interested in, the Hausdorff distance also makes them a metric space, as mentioned in \Cref{sec:preliminaries}.
    \item 
    We say $B$ is a \textit{basis} of a topological space $X$, if every open set $o$ is the union of sets from $B$.
    For example, all balls of radius $r<1$ form a basis of $\R^n$.
    A space $X$ is called \textit{locally contractible} if it has a basis of open subsets that are contractible.
    Again this holds for $\R^n$ and all of its subspaces as balls are contractible.
\end{enumerate}

The above considerations show that the preconditions of Smale's theorem are satisfied.
We thus get that $f$ is a weak homotopy equivalence.

For the following, we need some language from homotopy theory.
Many of the definitions are rather lengthy, and stating them here would exceed the scope of this paper.
We therefore refer the reader to any book on algebraic topology, for example Hatcher's excellent book \cite{Hatcher}.
The main conclusion of the following can however be easily summarized: \emph{for cubical complexes, weak homotopy equivalence implies homotopy equivalence}.

A map $f:X\rightarrow Y$ is a \textit{weak homotopy equivalence} if for every $n\geq 0$ it induces isomorphisms $f_*:\pi_n(X)\rightarrow \pi_n(Y)$, where $f_*$ denotes the \textit{induced homomorphism} of $f$ and $\pi_n$ denotes the $n$'th homotopy group.
This can be seen as an algebraic formulation of homotopy equivalence, but it is in general weaker.
However, for many spaces, weak homotopy equivalence implies homotopy equivalence.
One important family of such spaces are CW complexes (also called cell complexes), as given by Whitehead's Theorem (see e.g.~Theorem 4.5 in \cite{Hatcher}):

\begin{theorem}[Whitehead's theorem]
If a map $f:X\rightarrow Y$ between connected CW complexes induces isomorphisms $f_*:\pi_n(X)\rightarrow \pi_n(Y)$ for all $n$, then $f$ is a homotopy equivalence.
\end{theorem}

Finally, cubical complexes are special cases of CW complexes, thus $f=\pi$ is indeed a homotopy equivalence.

\section{NP-Membership of \textsc{Point-Vertex Art Gallery}}\label{app:NPcontainment}
In this section we prove that  \textsc{Point-Vertex Art Gallery} is in NP.
Note that the result can be considered folklore.
For the sake of completeness, we provide a self-contained proof.

\begin{lemma}
Given a simple polygon $P$ and an integer $k$, it can be decided in non-deterministic polynomial time whether the vertices of $P$ can be guarded by $k$ guards positioned in $P$.
\end{lemma}
\begin{proof}
    In order to show NP-membership, we show that for each polygon guardable by $k$ guards there exists a solution where all guards have small coordinate size. Such a guard placement can clearly be verified in polynomial time.
    Let $P$ be a (simple) polygon and
    let $G$ be an arbitrary solution with $k$ guards.
    We will show that there exists some set of guards $G'$
    with $|G| = |G'|$ and all guards in $G'$ have small coordinates.
    As a first step, we compute the visibility polygon for each vertex of $P$.
    This can be computed in polynomial time~\cite{abrahamsen2013constant, LinearVisibilityRegion, VisibilityHoles}.
    See \Cref{fig:NP-membership}, for an illustration.
    Now, the boundary of all of those visibility regions forms an arrangement \A of quadratic size. 
    Note that by definition in each cell of the arrangement \A every point sees exactly the same set of vertices of $P$.
    Furthermore, every vertex of \A has polynomial size coordinates, as it is the intersection of some lines that are defined by the vertices of the polygon $P$.
    Finally, for each $g\in G$, we pick a vertex of the cell of \A that contains $g$, and add it to $G'$.
    Clearly, as $G$ guards the vertices of $P$, so does $G'$. 
    Furthermore, by definition the size of $G$ and $G'$ are the same.
    This shows the claim and finishes the proof of the lemma.
\end{proof}

\begin{figure}
    \centering
    \includegraphics{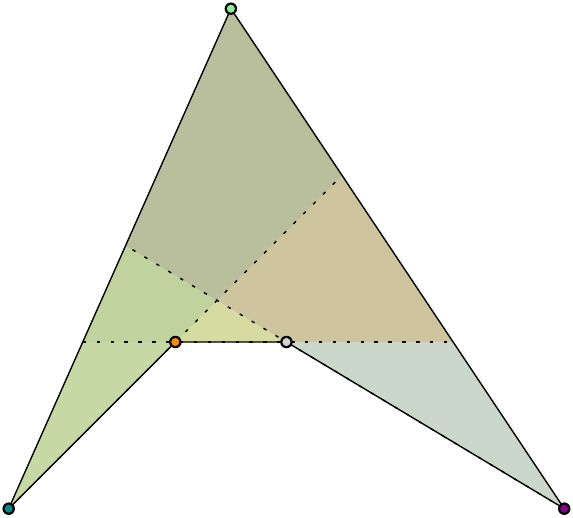}
    \caption{The polygon $P$ together with the visibility regions of each vertex of the polygon $P$.}
    \label{fig:NP-membership}
\end{figure}
\end{document}